\newif\ifnarrowcols

\ifnarrowcols
\documentclass[10pt, conference, compsocconf]{IEEEfocs}
\IEEEoverridecommandlockouts
\else
\documentclass[11pt,letterpaper]{article}
\fi

\usepackage[square,numbers]{natbib}
\usepackage{amsmath}
\usepackage{amssymb}
\usepackage{latexsym}
\usepackage{epsfig}
\usepackage{color}
\usepackage{enumerate}
\usepackage[margin=1in]{geometry}
\usepackage{times}

\ifnarrowcols
\else
\fi

\newtheorem{definition}{Definition}
\newtheorem{theorem}{Theorem}
\newtheorem{lemma}{Lemma}

\newtheorem{observation}{Observation}
\newtheorem{proposition}{Proposition}

\newtheorem{corollary}{Corollary}

\newtheorem{notation}{Notation}

\newlength{\boxwidth}
\makeatletter
\DeclareRobustCommand{\qed}{%
  \ifmmode 
  \else \leavevmode\unskip\penalty9999 \hbox{}\nobreak\hfill
  \fi
  \quad\hbox{\qedsymbol}}
\newcommand{\openbox}{\leavevmode
  \hbox to.77778em{%
  \hfil\vrule \vbox to.675em{\hrule width.6em\vfil\hrule}%
  \vrule\hfil}}
\newcommand{\qedsymbol}{\openbox}

\newenvironment{proof}[1][\proofname]{\par \normalfont
  \topsep6\p@\@plus6\p@ \trivlist
  \item[\hskip\labelsep\bfseries\itshape #1.]\ignorespaces }{%
  \qed\endtrivlist }

\newcommand{\proofname}{Proof}

\newcommand{\nset}{\mbox{{\normalfont I\hspace{-.4ex}N}}}

\newcommand{\hide}[1]{}

\newcommand{\F}{{\ensuremath {\mathcal F}}}
\newcommand{\G}{{\ensuremath {\mathcal G}}}
\newcommand{\GG}{{\ensuremath {\mathcal G}{\mathcal G}}}
\newcommand{\K}{{\ensuremath {\mathcal K}}}
\newcommand{\NE}{{\ensuremath {\mathcal N}}}
\renewcommand{\S}{{\ensuremath {\mathcal S}}}

\newcommand{\cube}{K}
\def\pp{{\bf p}}
\def\payoff{\mbox{\it payoff}}
\def\PSPACE{\mbox{\rm PSPACE}}
\def\switch{\it switch}
\def\brmapfn{Brouwer-mapping function}
\def\bmf{bmf}

\title{The Complexity of the Homotopy Method, Equilibrium
Selection, and Lemke-Howson Solutions}

\ifnarrowcols
\author{\IEEEauthorblockN{Paul W. Goldberg, Rahul Savani}
\IEEEauthorblockA{Dept.\ of Computer Science\\
University of Liverpool,\\
Liverpool L69 3BX, UK.\\
Email: goldberg@liv.ac.uk, rsjs@liv.ac.uk}
\and
\IEEEauthorblockN{Christos H. Papadimitriou}
\IEEEauthorblockA{Computer Science Division, Soda Hall,\\
University of California at Berkeley,\\
Berkeley, CA 94720, USA.\\
Email: christos@cs.berkeley.edu}
}
\else
\author{
Paul W.\ Goldberg\thanks{Supported by EPSRC Grant EP/G069239/1 ``Efficient Decentralised Approaches in Algorithmic Game Theory''}\\
Dept.\ of Computer Science\\
University of Liverpool,\\
Ashton Street, Liverpool L69 3BX, UK.\\
{\tt P.W.Goldberg@liverpool.ac.uk}
\and
Christos H. Papadimitriou\\
University of California at Berkeley,\\
Computer Science Division, Soda Hall,\\
Berkeley, CA 94720, USA.\\
{\tt christos@cs.berkeley.edu}
\and
Rahul Savani\\
Dept.\ of Computer Science\\
University of Liverpool,\\
Ashton Street, Liverpool L69 3BX, UK.\\
{\tt Rahul.Savani@liverpool.ac.uk}
}
\fi

\begin{document}

\maketitle

\begin{abstract}
We show that the widely used homotopy method for solving fixpoint
problems, as well as the Harsanyi-Selten equilibrium selection
process for games, are \PSPACE-complete to implement.
Extending our result for the Harsanyi-Selten process, we show that 
several other homotopy-based algorithms for finding equilibria of games  
are also \PSPACE-complete to implement.
A further application of our techniques yields the result that it is
\PSPACE-complete to compute any of the equilibria that could be found
via the classical Lemke-Howson algorithm, a complexity-theoretic 
strengthening of the result in \cite{SvS}.
These results show that our techniques can be widely applied and 
suggest that the \PSPACE-completeness of implementing 
homotopy methods is a general principle.
\end{abstract}

\ifnarrowcols
\else
\newpage
\fi

\section{Introduction}
According to Roger Myerson \cite{My}, the 1950 publication of Nash's
paper on equilibria was a watershed event not just for Game Theory,
but for Economics in general. The new general equilibrium concept,
and its established universality, was an impetus for understanding
rationality in much more general economic contexts, and inspired the
important price equilibrium results by Arrow and Debreu. Myerson
argues convincingly in \cite{My} that the concept of Nash
equilibrium lies at the foundations of modern economic thought.

Seen from an algorithmic perspective, however, the Nash equilibrium
suffers from two important problems: First, it is not clear how to
find it efficiently (the same is true for the Arrow-Debreu variety
for markets and prices). This shortcoming had already been
identified by economists since the 1950s, and much effort has been
devoted to algorithms for finding Nash equilibria, see
\cite{SH,LH,HP1} for examples from a very extensive literature. None
of these algorithms came with polynomial-time guarantees, however,
and the recent result \cite{DGP,CDT} establishing that the problem is
PPAD-complete explains why. Of the many algorithmic approaches
proposed by economists over the past 50 years for finding Nash
equilibria, most have been shown by now to require exponential time
in the worst case \cite{SvS,HPV}. 
One exception is an important
algorithmic genre known as {\em homotopy methods} \cite{Ea}; see
\cite{HP} for a recent survey. 

In topology, a {\em homotopy} is a
continuous transformation from one function to another (as, for
example, between two paths joining two points on a map). The
homotopy method starts with a fixpoint problem that is easy to solve
(say, a rotation of a disc around its center), and continuously
transforms it into the problem in hand, by ``pivoting'' to new
fixpoints along the way.  A theorem by Browder \cite{B} establishes
the validity of this method in the limit, by showing the existence
of a continuous path of fixpoints that joins two fixpoints of the
initial and the final problems.

The second algorithmic obstacle for the Nash equilibrium concept is {\em
multiplicity.} Games have multiple equilibria, and markets many
price equilibria, and thus the corresponding equilibrium concepts
are only {\em nondeterministic predictions} (oxymoron intended). In
price equilibria, this multiplicity has been blamed for {\em
economic crises:} The path guaranteed by Browder's theorem is
non-monotonic, going back and forth in time. As a result, equilibria
vanish at its folds, leaving the market in turmoil~\cite{Ba}. In
games, a proposed remedy for multiplicity is the so-called focal
point theory see, e.g., \cite{Kr} p.~414, postulating that players
implicitly coordinate their equilibrium choice by focusing on the
most obvious, or mutually advantageous, equilibrium; repeated play
and learning (see, e.g., \cite{FL}) can also be considered a remedy
for multiplicity. In 1975, Harsanyi proposed the {\em tracing
procedure} \cite{Harsanyi} for battling equilibrium multiplicity, a theory
further explicated in his joint 1988 book ``A General Theory of
Equilibrium Selection in Games'' with Selten \cite{HS} (Harsanyi and
Selten shared in 1994 the Nobel prize with Nash). The tracing
procedure asserts that players engaged in a game $\G$ play at first
a simple game $\G_0$, in which their prior beliefs about the other
players' behavior result in a dominant strategy. As time $t$
progresses, and their priors are falsified by life, they play a more
and more realistic game $\G_t = (1-t)\cdotp \G_0+t\cdotp \G$, until,
at time $t=1$, they end up playing the intended game $\G$. They show
that, for almost all games, tracing the equilibrium path of this
process results in a unique equilibrium. Notice the parallel with
the homotopy method; apparently the two were discovered
independently.

\paragraph{Our results\ifnarrowcols\else{.}\fi}
{\em This paper is a complexity-theoretic critique of the tracing
procedure and the homotopy method:} we show that
finding the solutions they prescribe requires the power 
of \PSPACE. In particular, finding the Brouwer
fixpoint that would have been discovered by the homotopy method, for
a simple starting function and an adversarial final one, is
\PSPACE-complete. The same is true, via standard reductions, for
price equilibria.  We also construct examples where the homotopy
method not only will undergo an exponential number of pivots (this
was expected since~\cite{HPV}), but will suffer an exponential
number of {\em direction reversals.}   As for the tracing procedure,
we show that it is \PSPACE-complete to find the Nash equilibrium
selected by it, even in two-player games, and even if the initial game
has dominant strategies obtained from priors, exactly as prescribed
by Harsanyi and Selten.  
We extend this result to homotopy-based algorithms where the starting 
game depends on the final game and show that it is 
\PSPACE-complete to implement the 
Herings-van den Elzen, Herings-Peeters, and van den Elzen-Talman algorithms 
for finding equilibria in games.
Finally, it is particularly noteworthy that
\PSPACE-completeness prevails even for finding the solutions that would be
returned by the classical Lemke-Howson algorithm, a simplex-like method that
had long been considered an oasis of conceptual simplicity
and (until \cite{SvS}) of algorithmic hope in this field.
This reinforces the ``exponentially long paths'' result
of~\cite{SvS} with a new result which says that, subject only to the hardness
of \PSPACE, no short cuts to Lemke-Howson solutions are possible (for any of the different 
initial choices of the algorithm).
Since it is known that the Lemke-Howson
algorithm can be expressed as a homotopy~\cite{HP}, this result can also
be seen as a powerful specialization of our first result.

The algorithms we consider solve problems in the complexity class PPAD, which
is contained in TFNP, the class of all total function problems in NP.  Another
prominent complexity class contained in TFNP is PLS (for polynomial local
search).  Many common problems in PLS (e.g., local max cut and finding pure
equilibria of congestion games) are complete under a so-called {\em tight}
PLS-reduction, implying that the corresponding standard local search algorithm
is exponential (for certain starting configurations and {\em any} choices of
the local search algorithm).   Furthermore, one can conclude that the
computational problem of finding a local optimum reachable from a given
starting configuration by  local search is \PSPACE-complete. 

No such concept of tight reductions is known for PPAD, and our results
can be seen as addressing this deficiency. 
Specifically, we show the \PSPACE-completeness (and exponential worst-case behaviour) of a number of homotopy-based
algorithms for finding equilibria.  Our reductions start with the problem {\sc
Other end of this line} ({\sc Oeotl}),
which is related to the problem {\sc End of the line} used in the definition is
PPAD, seeking not just any end of a path, but the other end of the particular
path starting at the origin.  {\sc Oeotl} was known to be \PSPACE-complete since~\cite{Pa}, but this fact has so
far remained unexploited for proving lower bounds for other problems.  

\paragraph{Outline of the paper\ifnarrowcols\else{.}\fi}
In Section~\ref{sec:prelim_hom}, we give an overview of the linear
homotopy method as applied to Brouwer functions and games.  In
Section~\ref{sec:com}, we recall the \PSPACE-complete problem {\sc Oeotl} ({\sc Other
End of this Line}), which serves as the starting point for all our main reductions. 
In Section~\ref{sec:hom}, we show that the linear homotopy method to
compute a Brouwer fixpoint is \PSPACE-complete, which is proved in
Section~\ref{sec:hom:red}. In Section~\ref{sec:lt}, we establish the
\PSPACE-completeness of the linear tracing procedure for two-player
strategic form games for a special starting game that is independent
of the final game.  These results are extended to starting games that
depend on the final game in Section~\ref{sec:extVdET-HP},
where we show that it is \PSPACE-complete to implement the 
Herings-van den Elzen, Herings-Peeters, and van den Elzen-Talman algorithms 
for computing equilibria of games.
The techniques of~\cite{CDT,DGP} are central to both Section~\ref{sec:hom}
and Section~\ref{sec:lt} and are recalled and extended along the way.
Finally, in Section~\ref{sec:extLH}, we show
that it is \PSPACE-complete to find any solution 
of a two-player game by the Lemke-Howson algorithm.

\section{Preliminaries}
\label{sec:prelim}

\subsection{Homotopies} \label{sec:prelim_hom}
A {\em Brouwer function} $\F$ is a continuous function from a convex and compact domain $D$ to itself;
by Brouwer's fixpoint theorem there exists $x\in D$ such that $\F(x)=x$.
A {\em homotopy} between two functions $\F_0:X\longrightarrow Y$ and $\F_1:X\longrightarrow Y$ (where
$X$ and $Y$ are topological spaces) is a continuous function $H:[0,1]\times X\longrightarrow Y$ such that
for all $x\in X$, $H(0,x)=\F_0(x)$ and $H(1,x)=\F_1(x)$. In this paper, we are interested in the special case
where $X=Y=D$, for $D$ a closed compact subset of Euclidean space, such as a cube.
Thus, $\F_0$ and $\F_1$ are Brouwer functions on $D$.
Given two continuous functions $\F_0,\F_1:D\longrightarrow D$, the {\em linear homotopy} is given by the expression
$H(t,x)=(1-t)\F_0(x)+t\F_1(x)$, and (if $D$ is convex) results in a continuum of Brouwer
functions $\F_t:D\longrightarrow D$ given by $\F_t=(1-t)\cdotp \F_0 + t\cdotp \F_1$ for $t\in[0,1]$.

Browder's fixpoint theorem \cite{B} (not to be confused with Brouwer's fixpoint theorem)
asserts that given a homotopy connecting $\F_0$ and $\F_1$,
there is a path in $[0,1]\times D$ from some fixpoint of $\F_0$ to some
fixpoint of $\F_1$, such that for every point $(t,x)$ on that path, $x$ is a fixpoint of $\F_t$.
The homotopy method \cite{Ea,HP} for finding a fixpoint of $\F_1$ selects $\F_0$ to have
a unique and easy to find fixpoint, and essentially follows such a path.
As noted in~\cite{HP}, we do not expect the path to be monotonic in $t$ --- indeed, we show in
\ifnarrowcols\else Section~\ref{sec:exp-changes} \fi
that an exponential number of direction reversals is possible.

We are often interested in {\em approximate} fixpoints\footnote{A very
interesting alternative consideration \cite{EY} focuses on exact
fixpoints, resulting in higher complexity of the search problem;
here we could also consider exact fixpoints and equilibria without
much effect on our results, since we are dealing with
\PSPACE-completeness. It is known from~\cite{EY} that this harder
problem belongs to \PSPACE.}. If $\F$ is a Brouwer function, an
$\epsilon$-approximate fixpoint is a point $x$ such that
$|\F(x)-x|\leq \epsilon$ (we shall use the $L_{\infty}$ metric
throughout).  It follows from Browder's theorem that, for any
$\F_0,\F_1$, there is a finite sequence
$x_0,x_{t_1},\ldots,x_{t_k},x_1$ of $\epsilon$-approximate fixpoints
of $\F_0,\F_{t_1},\ldots,\F_{t_k},\F_1$, for some $k$ and
$t_1,\ldots,t_k$, such that any two consecutive fixpoints in the
sequence are at most $\epsilon$ apart.

We shall be interested in the following problem, which we call {\sc Browder fixpoint}: 
Given two arithmetic circuits computing two
functions $\F_0$ and $\F_1$ from $[0,1]^d$ to itself with Lipschitz
constant $\ell$, an $\epsilon>0$, where $\F_0$ has a unique fixpoint $x_0$,
find an $\epsilon$-approximate fixpoint $x_1$ of $\F_1$ that is connected via
a sequence of $\epsilon$-approximate fixpoints to $x_0$. (To make this
definition precise, we of course have to identify classes of functions from which
$\F_0$ and $\F_1$ may be drawn.) Notice that the homotopy
method for computing Brouwer fixpoints provides a solution to this problem.

Homotopies can be defined very similarly also for games.  Given two
games $\G_0,\G_1$ of the same type (number of players and
strategies), we consider $\G_t=(1-t)\cdotp \G_0 + t\cdotp \G_1$, where
it is the players' utilities that are interpolated.  It is routine
to extend this definition to more general classes of games, such as
graphical games~\cite{K} (in which case, in addition to the players and
strategies, the two graphs must be the same). Browder's theorem, via
Nash's reduction, establishes that there is a path of approximate
Nash equilibria here as well. The problem {\sc Linear tracing} is
the following: Given two games $\G_0$ and $\G_1$, an $\epsilon>0$,
and a Nash equilibrium $x_0$ of $\G_0$, find an $\epsilon$-approximate Nash
equilibrium $x_1$ of $\G_1$ that is connected via a sequence of
$\epsilon$-approximate Nash equilibria to $x_0$.

It is easy to see that {\sc Linear tracing} is in \PSPACE, and it can
be checked that the algorithm of Herings and van den Elzen~\cite{HE}
achieves this\ifnarrowcols\else(see Section~\ref{sec:lt:membership})\fi.
{\sc Browder Fixpoint} is also in \PSPACE.

\subsection{\bf \sc Other End Of This Line} \label{sec:com} 
We consider directed graphs on $2^n$ vertices represented as $n$-bit
vectors. The arcs are represented by two polynomial-size circuits $S$ and
$P$, each having $n$ inputs and outputs, as follows. There is an arc from vertex
$v$ to $w$ provided that $S(v)=w$ and $P(w)=v$. Notice that all vertices
of the graph have both indegree and outdegree $0$ or $1$, that is, the
graph consists of paths, cycles, and isolated vertices.

\begin{definition}
An $(S,P)$-graph with parameter $n$ is a graph on $\{0,1\}^n$ specified by
circuits $S$ and $P$, as described above, subject to the constraint that 
vertex $0^n$ has no incoming arc but does have an outgoing arc.
\end{definition}

The problem {\sc End of the line} is the problem of finding a vertex of a given
$(S,P)$-graph other than $0^n$ which has at most one incident arc. Note that this
problem is in the class TFNP of total search problems in NP: there exists a
solution that could be obtained by following the directed path that starts at $0^n$, and
any given solution may be efficiently checked for correctness.
The class PPAD~\cite{Pa} is defined as all search problems polynomial-time reducible to
{\sc End of the line.} The problem {\sc Other end of this line} (which we will
subsequently abbreviate to {\sc Oeotl}) is the problem of finding
the end of {\em the particular path} that starts at $0^n$. In contrast with
{\sc End of the line}, a given solution to an instance of {\sc Oeotl} has
no obvious concise certificate that it is the correct endpoint, so while {\sc Oeotl} is a
total search problem, it is apparently not an NP total search
problem. In fact, we have the following (Theorem 2 of~\cite{Pa}),
which is the starting-point of our reductions.

\begin{theorem}\label{thm:oeotl}~{\cite{Pa}}  {\sc Oeotl} is \PSPACE-complete.
\end{theorem}

\ifnarrowcols
\else
\subsection{Our Approach}

In~\cite{DGP}, each instance $I$ of {\sc End of the line} is reduced
to a game $\G_I$ in such a way that any Nash equilibrium of $\G_I$
efficiently encodes a solution to $I$. Here we reduce $I$ to a
homotopy problem defined by two games, $(\G_0,\G'_I)$, where $\G_0$
depends only on $n$, while $\G'_I$ encodes the circuits in $I$, and
is an extension of $\G_I$ of~\cite{DGP}. We establish that the
associated linear homotopy corresponds to the naive ``follow the
line'' approach to solving {\sc Oeotl}; technically, a suffix of the
homotopy path corresponds to following the line, and the existence of
the relevant suffix is established in a non-constructive way, using
the intermediate value theorem. In extending the result to the
Lemke-Howson algorithm, the main technical obstacle is the initial
choice of which ``label to drop'', leading to multiple disjoint paths
in the mixed-strategy profile space. We have to ensure that all of the
$2n$ solutions (one for each pure strategy) efficiently encode the
solution to $I$, where $I$ is treated as an instance of {\sc Oeotl}.
This is done by embedding two copies of the game $\G'_I$ inside a
larger one in such a way that at least one copy does not contain the
initially-dropped label, and arguing 
that any Lemke-Howson equilibrium restricted to this copy ends up
encoding the unique solution to~$I$.
\fi

\section{The Homotopy Method for Brouwer Fixpoints}\label{sec:hom}

In this section we give detailed definitions of classes of fixpoint and approximate
fixpoint computation problems. In Section~\ref{sec:hom:def}, we review the definition
of {\brmapfn}s\ ---and related concepts--- from Chen et al.~\cite{CDT}, here applied
to a three dimensional domain.  In Section~\ref{sec:hom:impl}, we review the techniques of~\cite{DGP,CDT}
for implementing {\brmapfn}s as arithmetic circuits.
In Section~\ref{sec:hom:red}, we prove Theorem~\ref{thm:pspace-brouwer}, the main result of
Section~\ref{sec:hom}, in which we establish the \PSPACE-completeness
of a linear homotopy for finding a fixpoint of a Brouwer function. 
$n\in\nset$ will denote a complexity parameter of problem instances.
We define a sequence $\F^{(n)}_0$  of ``basic Brouwer functions'' having unique known fixpoints.
For each $n$ we define a class of Brouwer functions whose members encode $(S,P)$-graphs on $\{0,1\}^n$.
The homotopy of Equation~(\ref{eq:homotopy}) defines a class of functions $\F_t$, $t\in[0,1]$,
that interpolate between $\F_0$ and $\F_1$ and specifies a particular fixpoint of $\F_1$.
We will show that from that fixpoint, we can efficiently recover a solution to {\sc Oeotl}
for the graph encoded by~$\F_1$.

\subsection{Definitions and notation}
\label{sec:hom:def}

\begin{notation}
Let $\cube$ be the unit 3-D cube $[0,1]^3$. For $n\in\nset$ let
$\K^{(n)}$ denote a partition of $\cube$ into $2^{3n}$ ``cubelets'',
$\K^{(n)}=\{\cube_{ijk}~:~0\leq i,j,k\leq 2^n-1\}$; $\cube_{ijk}$ 
is an axis-aligned cube of length $2^{-n}$ whose vertex closest to the
origin has coordinates $2^{-n}(i,j,k)$.  
\end{notation}

We define a {\em Brouwer-mapping circuit} in a similar way to the
definition in~\cite{CDT}, here specialized to the case of 3 dimensions.
We also introduce some variations of the definition, as follows:
\begin{definition}\label{def:bmf}
\ifnarrowcols \else (Brouwer-mapping circuit/function; basic \brmapfn; DGP-style \brmapfn; partial \brmapfn) \fi

A {\em Brouwer-mapping circuit (bmc)} is a directed boolean circuit with
$3n$ input nodes and $2$ output nodes. Note that any bmc $B$ has an associated
{\em \brmapfn\ (\bmf)} $f_B:\K^{(n)}\longrightarrow\{0,1,2,3\}$ 
that maps any cubelet $\cube_{ijk}$ to one of the four colors $\{0,1,2,3\}$.
We require the colors of all exterior cubelets to be predetermined as follows.
For $i=0$, $f_B(\cube_{ijk})=1$.
For $j=0$, $i>0$, $f_B(\cube_{ijk})=2$.
For $k=0$, $i,j>0$, $f_B(\cube_{ijk})=3$.
All other exterior cubelets are mapped to $0$.

The {\em basic \bmf} $f^{(n)}_0:\K^{(n)}\longrightarrow\{0,1,2,3\}$
has the additional property that all internal cubelets get mapped to 0. Notice
that $f^{(n)}_0$ is computable by a bmc of size polynomial in $n$.

A {\em DGP-style \bmf} is one that is derived from an $(S,P)$-graph in
the manner of~\cite{DGP}, and so is computable with a bmc of size polynomial
in the size of circuits $S$ and $P$. (Proposition~\ref{prop:graph} notes the relevant
property of DGP-style \bmf's.)

A {\em partial \bmf} $f$ is defined with respect to a
set $\S \subseteq \K^{(n)}$; $f$ assigns a color
to elements of $\S$ but $f$ may be undefined on non-elements of $\S$.
\end{definition}

\begin{proposition}\label{prop:graph}
The following problem is \PSPACE-complete. Given a Brouwer-mapping circuit $B$, find a point in $\cube$
that is a vertex of 4 cubelets mapped to all 4 colors by the associated {\bmf} $f_B$,
and which is connected to the origin via cubelets having colors other than 0.
\end{proposition}

This is a total search problem: the topological intuition is that there is a
line that is adjacent to the colors $\{1,2,3\}$ and has one end at $2^{-n}(1,1,0)$.
The other end must be inside $\cube$ and adjacent to color 0, since no other exterior
point is adjacent to the colors $\{1,2,3\}$.
\ifnarrowcols \else We note in passing that if we did not make
the ``connected to the origin'' requirement, the problem would be PPAD-complete. \fi

\ifnarrowcols
\else
\begin{proof}
We reduce from {\sc Oeotl} (Theorem~\ref{thm:oeotl}): Let $G$ be an $(S,P)$-graph.
Let $f_B$ be a DGP-style \bmf\ derived from $G$, whose circuit $B$ efficiently encodes $G$.

Given a \bmf\ $f_B$, define a {\em $\{1,2,3\}$-chromatic vertex} to be one that is shared
by 3 cubelets with colors 1, 2 and 3. By construction, the only exterior $\{1,2,3\}$-chromatic
vertex for any \bmf\ is $2^{-n}(1,1,0)$. Form a digraph $G_B$ on $\{1,2,3\}$-chromatic vertices by
adding an arc between any pair that share a cubelet, directed such that if it pointing away from
a viewer, its adjacent colors $1$, $2$, $3$ will appear in clockwise order around it.
The reduction of~\cite{DGP} ensures $G_B$ has indegree/outdegree at most 1.

Define a {\em panchromatic vertex} to be one that belongs to 4 cubelets of all 4 different colors.
By construction, for all \bmf's there is a path of $\{1,2,3\}$-chromatic vertices starting
at $2^{-n}(1,1,0)$ and ending at a unique panchromatic vertex $v_{end}$. $v_{end}$ is a
solution; it can be found in polynomial space by following this path. 

Let $G_B$ be the graph on $f_B$'s $\{1,2,3\}$-chromatic vertices as described above.
The reduction of~\cite{DGP} has the following properties, from which the result follows.
\begin{enumerate}
\item Each vertex $v$ of  $G$ has an associated $\{1,2,3\}$-chromatic vertex $b(v)$ of $G_B$;
$v$ and $b(v)$ may be computed in polynomial time from each other. For $v=0^n$,
$b(v)=2^{-n}(1,1,0)$.
\item $v$ is a solution to {\sc End of the line} if and only if $b(v)$ is panchromatic.
\item Each arc $(v,w)$ of $G$ corresponds to a sequence of edges of $G_B$ that connect
$b(v)$ to $b(w)$.
\item Each connected component of cubelets colored with $\{1,2,3\}$, corresponds to
a connected component of $G$.
\end{enumerate}
\end{proof}
\fi

\subsection{Implementing Brouwer-mapping functions as arithmetic circuits}
\label{sec:hom:impl}

We review a class of functions used to establish PPAD-completeness of graphical
and strategic-form games. Recall that $\cube$ denotes the 3-dimensional unit cube;
we consider continuous functions $\F:\cube\longrightarrow \cube$
having the following structure.  Each function is an arithmetic circuit composed of
nodes, with each node taking inputs from up to 2 other nodes, and producing
an output, for example, the sum of its inputs. All values are constrained to $[0,1]$, so a
node that adds its inputs would output 1 if their sum is greater than 1.  Identify 3 nodes
as ``input nodes'' and another 3 as ``output nodes'', so if~$\F$ is a continuous function
from $\cube$ to $\cube$, it has a Brouwer fixpoint.

\begin{definition}\label{def:circuit}
A {\em linear arithmetic circuit} is an arithmetic circuit that computes a function
from $\cube$ to $\cube$, represented by a directed graph whose nodes are ``gates'' that
perform certain basic arithmetic operations on their inputs as follows.
Each gate takes as input 0, 1 or 2 real values in $[0,1]$ and outputs a single real value
in $[0,1]$, where the output of a gate may be the sum/difference/max/min
of two inputs, or a constant multiple of a single input, or no input and
constant output. (An output value is set to 1 if for example two inputs
that sum to more than 1 are input to a ``sum'' gate.) We also allow ``comparator
gates'' in which the output of such a gate evaluates to 1 (respectively, 0) if its first input is
greater (respectively, less) than the second input, and may take any value
if they are equal.
\end{definition}

\begin{notation}
Let $\alpha=2^{-2n}$. Let $\delta_1=(\alpha,0,0)$, $\delta_2=(0,\alpha,0)$,
$\delta_3=(0,0,\alpha)$, $\delta_0=(-\alpha,-\alpha,-\alpha)$.
\end{notation}

\begin{definition}\label{def:impl}
We shall say that a \brmapfn\ $f$ is {\em implemented by an arithmetic
circuit $C$} if whenever $f(\cube_{ijk})=c$, then $C(x)-x = \delta_c$ when $x$ is at the center of
$\cube_{ijk}$. For $x$ not at a center, $C(x)-x$ should be a convex
combination of values of $C(z)-z$ for cubelet centers $z$ within
$L_\infty$ distance $2^{-n}$ of $x$. Given $\F:\cube\longrightarrow\cube$
computed by such a $C$, we shall similarly say that $\F$ {\em implements} $f$.
\end{definition}

\begin{observation}\label{obs:impl}
If $\F$ implements $f$, then any fixpoints of $\F$ must lie within distance
$2^{-n}$ of panchromatic vertices of $f$, and vice versa.
\end{observation}

\begin{theorem}\label{thm:bmf2circ}
A Brouwer-mapping function having complexity parameter $n$ can be implemented
using a linear arithmetic circuit having $poly(n)$ gates, that computes a
continuous function.
\end{theorem}

The proof gives a new technique to implement any Brouwer-mapping function $f$ as a {\em continuous}
function $\F$ that uses a linear arithmetic circuit. This is in
contrast with the corresponding techniques of~\cite{CDT,DGP} that used a
sampling-based approach in order to smooth the transition between distinct cubelets.
The sampling-based approach results in discontinuous functions, where
Browder's theorem would not be applicable (although it could still be applied to
a continuous approximation). The technique only works in constant dimension;
if can be extended to higher dimension using the ``snake-embeddings'' of~\cite{CDT}.
\ifnarrowcols \else (See Section~\ref{sec:polyerror}.) \fi
\ifnarrowcols \else The general idea of the technique is to take a simplicial decomposition of
the domain $\cube$, give rules for obtaining the values of $\F$ at the vertices
of the decomposition, and linearly interpolate within each simplex. \fi

\begin{proof}
Let $f:\K^{(n)}\longrightarrow\{0,1,2,3\}$ be a Brouwer-mapping function.
We construct a continuous Brouwer function $\F:\cube\longrightarrow\cube$
computed by a linear arithmetic circuit $C$ as follows.

For $x$ at the center of cubelet $\cube_{ijk}$, set $\F(x)-x = \delta_c$ where $c=f(\cube_{ijk})$.
For $x$ a vertex of cubelets $\K_x \subset \K^{(n)}$, set $\F(x)-x$ to be the average of $\F(z)-z$
for all points $z$ at the centers of members of $\K_x$. The relevant points $z$ can be obtained using a
polynomial-sized piece of circuitry.

Let $\S$ be a simplicial decomposition of the unit cube consisting of 12 simplices
that share a vertex at the center of the cube, and all other vertices are vertices of
the cube. Let $\S_{ijk}$ be the simplicial decomposition of cubelet $\cube_{ijk}$
obtained by scaling $\S$ down to $\cube_{ijk}$. Applied to all cubelets
in $\K^{(n)}$ this results in a
highly regular decomposition $\S^{(n)}$ of $\cube$ into $12.2^{3n}$ simplices.

For any $x\in\cube$, $F(x)$ is obtained by linearly interpolating between
the vertices of the simplex in $\S^{(n)}$ that contains $x$. Clearly $\F$ is continuous.

The result follows from the following claim:
\begin{proposition}
$\F$ as defined above, may be computed by a linear arithmetic circuit of size polynomial in $n$.
\end{proposition}

\begin{proof}
If $x$ is not a vertex of $\S^{(n)}$, the circuit can determine the vertices
of a simplex $S_x\in \S^{(n)}$ that contains $x$.
There may be more than one such simplex, in which case {\em it does not matter which is chosen}.

The circuit has 12 cases to consider, depending on the orientation of $S_x$. Each case can be
handled in the same general manner, by subtracting some vertex $v$ of $S_x$ from $x$, and multiplying
$(x-v)$ by some constants (the coefficients of the linear function that interpolated between
the vertices of $S_x$). Note that we never need to multiply two computed
quantities together, multiplication only ever takes place between a computed quantity and
a constant, as required for a linear arithmetic circuit.
\end{proof}
\end{proof}

\subsection{The \PSPACE\ reduction to linear arithmetic circuits}\label{sec:hom:red}

In this subsection, we establish the \PSPACE-completeness of the problem {\sc Browder fixpoint},
mentioned in the Introduction, which can now be made precise as follows.
We use two {\bmf}s $f_0$ and $f_1$, where $f_0$ is the {\em basic \bmf} of Definition~\ref{def:bmf},
and $f_1$ shall be a DGP-style {\bmf} that encodes an instance of {\sc End of the line} as
constructed in~\cite{DGP}.
Let $\F_0$ and $\F_1$ be implementations of $f_0$ and $f_1$ using linear
arithmetic circuits as described in the proof of Theorem~\ref{thm:bmf2circ}.
For $\F:\cube\longrightarrow\cube$ let $\F^{(i)}$ denote the $i$-th component of $\F$. For $i=1,2,3$ let
\begin{equation}\label{eq:homotopy}
\ifnarrowcols
\begin{tabular}{@{}r@{}c@{}l}
\else
\begin{tabular}{rcl}
\fi
$\bar{\F}_t^{(i)}$  & = & $(\F_0^{(i)}-t) + (\F_1^{(i)}-(1-t))$ \\
$\F_t^{(i)}$        & = & $\max(\min(\F^{(i)}_0,\F^{(i)}_1), \bar{\F}_t^{(i)})$
\end{tabular}
\end{equation}
where in~(\ref{eq:homotopy}), the outputs of operators $+$ and $-$ are restricted
to lie in $[0,1]$ (so, rounding to 0 or 1 if needed). $\F_t$ interpolates continuously
between $\F_0$ and $\F_1$ and is constructed from them using elements of the linear
arithmetic circuits of Definition~\ref{def:circuit} (which is useful later; the natural alternative
$\F_t = t\F_0 + (1-t)\F_1$ does not have this property.)

\begin{observation}\label{obs:cts}
For all $t\in[0,1]$, $\F^{(i)}_t$ is Lipschitz continuous, with Lipschitz value $<2.2^{-n}$.
\end{observation}

$\F_0$ has a unique fixpoint close to $2^{-n}(1,1,1)$. $\F_0$ is a ``basic Brouwer function''
which forms the starting-point of homotopies we consider. Hence Observation~\ref{obs:cts} and
Browder's fixpoint theorem implicitly define a corresponding fixpoint of $\F_1$.

Define an  {\em approximate fixpoint} of $\F:\cube\longrightarrow\cube$ to be
a point $x\in\cube$ with $|\F(x)-x|\leq \alpha/5$ (recall $\alpha=2^{-2n}$).

\begin{theorem}\label{thm:pspace-brouwer}
It is \PSPACE-complete to find, within accuracy $2^{-n}$, the coordinates
of the fixpoint of $\F_1$ that corresponds to the homotopy of~(\ref{eq:homotopy}).
It is also \PSPACE-complete to find the coordinates of an approximate fixpoint of $\F_1$
that would be obtained by following a sequence of approximate fixpoints
of $\F_t$ in which consecutive points are within distance $\alpha$ of each other.
\end{theorem}

\begin{proof}
We reduce from the problem defined in Proposition~\ref{prop:graph} as follows.
Let $B$ be a Brouwer-mapping circuit derived from {\sc Oeotl}-instance $(S,P)$ using Proposition~\ref{prop:graph}
and let $f_B:\K^{(n)}\longrightarrow\{0,1,2,3\}$ be the function computed by $B$.
Let $\F_1:\cube\longrightarrow\cube$
be the function computed by a linear arithmetic circuit that implements $f_B$,
and $\F_0$ be computed by a circuit that implements the basic \bmf\ $f_0$
(where both implementations apply Theorem~\ref{thm:bmf2circ}). $\F_t$ is given by~(\ref{eq:homotopy}).

Let $P$ be a connected subset of $\cube\times[0,1]$ such that for any $(x,t)\in P$, $x$ is a fixpoint
of $\F_t$, and $P$ contains $x_0\in(\cube,0)$ and $x_1\in(\cube,1)$. Browder's fixpoint
theorem (with Observation~\ref{obs:cts}) assures us that such a $P$ exists. We claim that $x_1$
is within distance $2^{-n}$ of the unique solution to $B$ of the problem specified in
Proposition~\ref{prop:graph} (and hence, given $x_1$ we can easily construct this solution).

Suppose otherwise. For $x_1$ to be a fixpoint (even an approximate one) of $\F_1$,
by Observation~\ref{obs:impl} it must be within distance $2^{-n}$ of a panchromatic vertex $v$ of $f_B$.
But now, $v$ is not connected to the origin via non-zero cubelets of $f_B$.
By connectivity of $P$, there must exist $(x,t)\in P$ such that $x$ lies within a cubelet $K_x$
where $f_B(K_x)=0$.

We may assume further that $x$ is at least $2^{-n}$ distant from any non-zero cubelet of $f_B$.
This follows provided we assume that connected components
of non-zero cubelets of $f_B$ are separated from each other by a layer of
0-colored cubelets of thickness at least 3. This may be safely assumed by increasing
$n$ by a factor of 3 and subdividing the cubelets. We note that
\begin{enumerate}
\item each entry of vector $\F_0(x)-x$ is $<-\alpha/5$, and
\item each entry of $\F_1(x)-x$ is $<-\alpha/5$.
\end{enumerate}
It follows that for $t\in[0,1]$, each entry of $f_t(x)-x$ is less than $-\alpha/5$,
since coordinatewise, $f_0 \leq f_t \leq f_B$. That means that $x$ cannot be
an approximate fixpoint of any $f_t$, contradicting the assumption as required.

Since $x$ is at least $2^{-n}$ distant from any non-zero cubelet of $f_B$, it is also
at least $2^{-n}$ distant from any non-zero cubelet of $f_0$, since for any cubelet
$K_{ijk}$, $f_B(K_{ijk})=0~\Longrightarrow~f_0(K_{ijk})=0$.
The implementation of any \bmf\ $f$ as a function $\F$ computed by a linear arithmetic circuit,
as referred to in Theorem~\ref{thm:bmf2circ}, ensures that
$\F(x)-x$ is a convex combination of vectors $\F(z)-z$ for cubelet centers $z$ in the
vicinity of $x$, and since all those cubelet centers are colored 0, we have that the
entries of $\F(x)$ are all less than $-\alpha/5$, as required.
\end{proof}

\section{The Linear Tracing Procedure}\label{sec:lt}

We now turn to games and Nash equilibrium. Let $\G$ denote an $n\times
n$ game that we wish to solve, assumed to be chosen by an adversary.
$\G_0$ is a game 
with a unique ``obvious'' solution. In $\G_0$ each player
receives payoff 1 for his first action, and payoff 0 for all others,
regardless of what the other player does.

\ifnarrowcols
\else
\begin{equation}\label{eqn1}
\G_0 =
\begin{tabular}{cccccc}
                  & \vline & $s^c_0$  & $s^c_1$   & $\ldots$ & $s^c_{n-1}$ \\ \hline
$s^r_0$           & \vline & $(1,1)$  & $(1,0)$   & $\ldots$ & $(1,0)$  \\
$s^r_1$           & \vline & $(0,1)$  & $(0,0)$   & $\ldots$ & $(0,0)$  \\
$\vdots$          & \vline & $\vdots$ & $\vdots$  &          & $\vdots$ \\
$s^r_{n-1}$       & \vline & $(0,1)$  & $(0,0)$   & $\ldots$ & $(0,0)$  \\
\end{tabular}
\end{equation}
\fi

In the problem {\sc Linear tracing} the solution consists of the Nash equilibrium
of $\G$ that is connected to the unique equilibrium $(s^r_0,s^c_0)$ of $\G_0$
via equilibria of convex combinations $(1-t)\G_0+t\G$.
We can also define an approximate version of this problem, where instances
include an additional parameter~$\epsilon$, and we seek an $\epsilon$-Nash
equilibrium that is connected to the solution of ${\cal G}_0$ via a
sequence of $\epsilon$-approximate solutions of ${\cal G}_t$. For the two-player case
we assume $\epsilon=0$. For more than 2 players, we need a positive $\epsilon$
to ensure that solutions can be written down as rational numbers.

\begin{theorem}\label{thm:lt}
{\sc Linear tracing} is \PSPACE-complete for 2-person games.
\end{theorem}

The same result then holds for strategic-form games with more than 2 players.
It holds for a value of $\epsilon$ that is exponentially small; we could again use the ideas
of~\cite{CDT} to obtain a version where $\epsilon$ is inverse polynomial.

\ifnarrowcols  \else
Our reduction uses the result of the previous section,
along with earlier reductions between strategic-form games and graphical
games. $\G_0$ has a similar role to the basic Brouwer function $\F_0$,
but the correspondence is indirect; generally $\F_0$ is associated with one of the
``intermediate games'' $\G_t$ for $t>0$.
\fi

\subsection{Brief overview of the proof ideas}

The following is a brief overview of the rest of this Section~\ref{sec:lt}.
Membership of \PSPACE\ can be deduced from~\cite{HE}. The reduction from the
\PSPACE-complete discrete Brouwer fixpoint problem of the previous section,
applies the idea from~\cite{DGP} of going via {\em graphical games} to
normal-form games.  We derive a type of graphical game in which a specific
player (denoted $v_{\switch}$) acts as a switch, allowing the remaining players to simulate
either the basic Brouwer-mapping function, or one associated with an
instance of the search for a discrete Brouwer fixpoint. $v_{\switch}$
governs this behavior via his choice of either one of two alternative
strategies, and we show that a {\em continuous} path of equilibria from one choice
to the other, results in an equilibrium that ultimately represents a solution
to {\sc Oeotl}. The graphical game is then encoded as a 2-player game such that
the linear-tracing procedure corresponds to this continuous path of
equilibria in the graphical game.

\ifnarrowcols  \else
\subsection{Membership of \PSPACE}\label{sec:lt:membership}

Herings and van den Elzen~\cite{HE} show how to find approximate
equilibria on multi-player games, implicitly constructing a degree-2
graph that has a vertex corresponding to $\NE_0$, the Nash equilibrium
of $\G_0$. Given a simplicial decomposition of $D\times [0,1]$ (where
$D$ is the space of mixed strategies of $\G$) vertices of the graph
correspond to simplices and subsimplices, and edges are implicitly 
defined by a lexicographical pivoting rule that governs a choice of
movement from simplex to adjacent simplex, at each step of the
algorithm. It can be checked that this algorithm establishes membership
of \PSPACE\ for multiplayer {\sc Linear tracing}.
\fi

\subsection{Graphical Games}

In a graphical game~\cite{K}, each player is a vertex of a graph, and his payoffs depend on his
own and his neighbors' actions. For a low-degree graph, this is one way that games
having many players may be represented concisely. A homotopy between two graphical
games $\GG_0$ and $\GG_1$ would require that these games have the same 
underlying graph, so that they differ only in their numerical
payoffs. In the graphical games considered here, each player has just
2 actions and 3 neighbors. The main result of this section is

\begin{proposition}\label{prop:gg}
Consider graphical games that contain a special player $v_{\switch}$ whose payoffs
are constant (unaffected by his own actions or the other players'). The following
problem is \PSPACE-complete: find a Nash equilibrium of the game where
$v_{\switch}$ plays 1, that is topologically connected to a Nash equilibrium in which
$v_{\switch}$ plays 0, via a path of Nash equilibria in which $v_{\switch}$ plays mixed strategies.
\end{proposition}

Let $\F_0$ and $\F_1$ be functions computed by linear arithmetic circuits
that implement Brouwer-mapping functions $f_0$ and $f_1$, where $f_0$ is the ``basic
\bmf'' of Definition~\ref{def:bmf}, and $f_1$ is a DGP-style \bmf\ that
encodes some instance of {\sc End of the line}.


\begin{notation}
In a graphical game in which all players have 2 pure strategies denoted
0 and 1, given a mixed-strategy profile for the players we let
$\pp[v]$ denote the probability that player $v$ plays 1.
\end{notation}

\begin{definition}\label{def:simulation}
\ifnarrowcols\else (Linear graphical game; simulation of bmf's and partial bmf's) \fi

Given a bmf $f$, we construct an associated graphical game $\GG_f$ as follows.
$\GG_f$ has 3 special players $(v_x,v_y,v_z)$ whose strategies $(\pp[v_x],\pp[v_y],\pp[v_z])$
represent a point in $\cube$. If $f$ is implemented by $\F:\cube\longrightarrow\cube$ we
use gadgets of~\cite{DGP} to simulate the nodes in the arithmetic circuit that computes
$\F$ (each node of the circuit has an additional associated player in $\GG_f$).
The game can pay them to adjust $(\pp[v_x],\pp[v_y],\pp[v_z])$ in the direction
$\F(\pp[v_x],\pp[v_y],\pp[v_z])$ $-(\pp[v_x],\pp[v_y],\pp[v_z])$.
Then the players $(v_x,v_y,v_z)$ are incentivized to play $\F(\pp[v_x],\pp[v_y],\pp[v_z])$.
Consequently a Nash equilibrium of $\GG_f$ corresponds to a fixpoint of $\F$.
Moreover, an $\epsilon$-Nash equilibrium corresponds to a ${\rm poly}(\epsilon)$-approximate fixpoint of $\F$.
We call $\GG_f$ a {\em linear graphical game} since we only allow players whose payoffs cause them to
simulate the gates of linear arithmetic circuits.

A game of the above kind is said to {\em simulate} $f$. We say further that a game $\GG$
simulates a partial bmf on a subset $S$ of cubelets, if for any $K\in S$, when
$(\pp[v_x],\pp[v_y],\pp[v_z])$ lie at the center of $K$ the players $(v_x,v_y,v_z)$
are incentivized to play $(\pp[v_x],\pp[v_y],\pp[v_z])+\delta_c$, where $c=f(K)$.
\end{definition}

\begin{lemma}~\label{lem:switch}
Given any linear graphical game $\GG_1$ that simulates a \brmapfn\ $f_1$,
we can efficiently construct a new game $\GG^+$ having a player $v_{switch}$
whose behavior can either cause $\GG$ to simulate $f_1$ (if $v_{switch}$ plays 1)
or cause $\GG$ to simulate $f_0$ if instead $v_{switch}$ plays 0.
\end{lemma}

$v_{\switch}$ shall serve as a ``switch'', in allowing the game to
switch between simulating $f_0$ and $f_1$ (using an additional 3 players
$(v^+_x,v^+_y,v^+_z)$ whose strategies represent a point in $\cube$) according to
whether $v_{\switch}$ plays 0 or 1. Of course, $v_{\switch}$ has a key role in the associated
two-player game.

\begin{proof}
For $i\in\{0,1\}$, let $\GG_i$ be a graphical game constructed from $f_i$
according to Definitions~\ref{def:bmf}, \ref{def:impl}, \ref{def:simulation}.
$\GG_i$ has 3 players/vertices whose mixed strategies, as represented
by the probabilities that they play 1, represent a point in $\cube$.
Denote these players $(v^i_x,v^i_y,v^i_z)$.

Construct a ``combined'' game $\GG^+$ as follows. $\GG^+$ contains all the
players in $\GG_0$ and $\GG_1$ together with a new player $v_{\switch}$, where
$v_{\switch}$ has the same fixed payoff for playing either 0 or 1.
We add 3 players $(v^+_x,v^+_y,v^+_z)$ whose mixed strategies represent a point in
$\cube$, and players $(\bar{v}^+_x,\bar{v}^+_y,\bar{v}^+_z)$, whose behavior is governed by
\begin{equation}\label{eq:gg-homotopy}
\ifnarrowcols \begin{tabular}{@{}r@{}c@{}l@{}}
\else \begin{tabular}{@{}r@{}c@{}l@{}} \fi
$\pp[\bar{v}^+_x]$ & = & $(\pp[v^0_x]-\pp[v_{\switch}])+(\pp[v^1_x]-(1-\pp[v_{\switch}]))$ \\
$\pp[v^+_x]$ & = & $\max(\pp[\bar{v}^+_x],\min(\pp[v^0_x],\pp[v^1_x]))$
\end{tabular}
\end{equation}
(and similar expressions for $v^+_y$ and $v^+_z$)
where the parentheses in the above expression are important since the outputs
of the operators $+$ and $-$ are truncated to lie in $[0,1]$.

Players from $\GG_0$ and $\GG_1$ that take input from nodes $v^0_i$ or $v^1_i$
respectively, are then modified to take that input from $v^+_i$ instead.
This completes the construction.
\end{proof}

\begin{proof} of Proposition~\ref{prop:gg}:
We reduce from the circuit homotopy of Theorem~\ref{thm:pspace-brouwer}.
Let $\{\F_t~:~t\in[0,1]\}$ be an instance of this circuit homotopy.
Construct $\G_1$ from $\F_1$ as per Definition~\ref{def:simulation}.
Construct $\GG^+$ as in Lemma~\ref{lem:switch}, and we make the following observation.
\begin{observation}\label{intermediate}
Suppose that in $\GG^+$ we have $\pp[v_{\switch}]=t\in(0,1)$. The resulting game
$\GG^+_t$ simulates a partial \brmapfn\ $f_t$ which is implemented by a Brouwer
function $\F_t$ that is (pointwise) a convex combination of $\F_0$ and $\F_1$
and is defined on the subset of cubelets where $f_0=f_1$.
Given a homotopy path of Nash equilibria of $\GG^+$ that start at the unique
equilibrium of $\GG^+$ that satisfies $\pp[v_{\switch}]=0$ and ends at an
equilibrium of $\GG^+$ in which $\pp[v_{\switch}]=1$, there is a corresponding
homotopy path from the fixpoint of $\F_0$ and a fixpoint of $\F_1$
(noting that~(\ref{eq:gg-homotopy}) is essentially the same as~(\ref{eq:homotopy})).
\end{observation}
That concludes the proof of Proposition~\ref{prop:gg}.\end{proof}

The following version of Lemma~\ref{lem:switch} is useful in the construction
for Lemke-Howson solutions, later on.

\begin{corollary}~\label{cor:switches}
Given any linear graphical game $\GG_1$ that simulates a \brmapfn\ $f_1$,
we can efficiently construct a new game $\GG^+$ having 2 players $v_{\switch}$ and $v'_{\switch}$
whose behavior can either cause $\GG^+$ to simulate $f_1$ (if both $v_{\switch}$, $v'_{\switch}$ play 1)
or cause $\GG^+$ to simulate $f_0$ if instead either or both play 0.
\end{corollary}

\ifnarrowcols \else
\begin{proof}
The proof of the previous Lemma is modified as follows.
We re-use Equation~(\ref{eq:gg-homotopy}) for players $(v^+_x,v^+_y,v^+_z)$.
The 3 players $(v^{++}_x,v^{++}_y,v^{++}_z)$ whose mixed strategies represent a point in
$\cube$ have behavior  governed by
\begin{equation}
\begin{tabular}{rcl}
$\pp[\bar{v}^{++}_x]$ & = & $(\pp[v^0_x]-\pp[v'_{\switch}])+(\pp[v^+_x]-(1-\pp[v'_{\switch}]))$ \\
$\pp[v^{++}_x]$ & = & $\max(\pp[\bar{v}^{++}_x],\min(\pp[v^0_x],\pp[v^+_x]))$
\end{tabular}
\end{equation}
again with similar expressions for $v^{++}_y$ and $v^{++}_z$.
\end{proof}
\fi

\subsection{From graphical to two-player strategic-form games}

In this subsection we prove the following theorem, from which Theorem~\ref{thm:lt}
follows\ifnarrowcols\else \ since we have previously noted membership of \PSPACE.\fi

\begin{theorem}\label{thm:ltdetails}
It is \PSPACE-hard to compute the Nash equilibrium of a given 2-player normal-form
game $\G_1$, that is obtained via the linear homotopy that starts from $\G_0$, a version
of $\G_1$ where the payoffs have been changed to give each player payoff 1 for his first
strategy and 0 for the others.
\end{theorem}

We reduce from the graphical game problem of Proposition~\ref{prop:gg}.
Let $\GG^+$ be a linear graphical game that includes a player $v_{switch}$ as per Proposition~\ref{prop:gg}.
First, modify $\GG^+$ to give $v_{switch}$ a small payment (say, $0.01$) to play 1, and zero to play 0.

We define a homotopy between two-player strategic-form games $\G_0$ and $\G_1$ such that equilibria of
$\G_1$ efficiently encode equilibria of $\GG^+$, and equilibria of $\G_t$ encode equilibria of versions
of $\GG^+$ where $v_{switch}$ has a bias towards playing 0.
We use the reduction of~\cite{DGP} (Section 6.1) from graphical games to
2-player games\ifnarrowcols\else \ (a similar reduction is used in~\cite{CDT} (Section 7) to
express {\em generalized circuits} (similar to our linear arithmetic circuits) as 2-player games).\fi

\ifnarrowcols\else In the context of \fi a mixed-strategy profile,
let $\Pr[s]$ denote the probability allocated to pure strategy $s$ by its player.

\begin{definition}\label{def:circgame}
A {\em circuit-encoding} 2-player game $\G$ has a corresponding graphical game
$\GG$ where the graph of $\GG$ is bipartite; denote it $G=(V_1\cup V_2,E)$;
each player (vertex) in $\GG$ has 2 actions (denote them 0 and 1) and payoffs that
depend on the behavior of 2 other players in the opposite side of $G$'s bipartition. Each vertex/action pair
$(v,a)$ of $\GG$ has a corresponding strategy in $\G$; for $v\in V_1$, $(v,a)$ belongs to the
row player and for $v\in V_2$, $(v,a)$ belongs to the column player. The payoffs in $\G$
are designed to ensure that in a Nash equilibrium of $\G$
\begin{itemize}
\item $\Pr[(v,0)]+\Pr[(v,1)] \geq 1/2n$ where $n$ is the number of players in $\GG$
\item if in $\GG$, $v$ plays 1 with probability $\Pr[(v,1)]/(\Pr[(v,0)]+\Pr[(v,1)])$ then
we have a Nash equilibrium of $\GG$.
\end{itemize}
\end{definition}

Let $\G$ be a circuit-encoding game derived from $\GG^+$ according to Definition~\ref{def:circgame}.
Associate $v_{\switch}$ with 2 strategies of of the column player of $\G$, and let $s^c_k$
and $s^c_{k+1}$ be these strategies. Hence a Nash equilibrium of $\G$ corresponds to one of
$\GG^+$ where the value $\pp[v_{\switch}]$ is given by the value $\Pr[s^c_{k+1}]/(\Pr[s^c_k]+\Pr[s^c_{k+1}])$.

\begin{observation}
If we take a circuit-encoding 2-player game, and award one of the players a small
bonus to play $(v,a)$, then this corresponds to incentivizing the player $v$ in
$\GG$ to select strategy $a$. The corresponding incentive for $v$ will be larger,
but only polynomially larger.
\end{observation}

Let $\G_0$ be a $(n+1)\times(n+1)$ game with strategies $\{s^r_0,\ldots
s^r_n\}$ for the row player, and $\{s^c_0,\ldots s^c_n\}$ for the column
player. Payoffs are as follows: each player receives 1 for playing $s^r_0$
or $s^c_0$, and 0 for $s^r_j$ or $s^c_j$ for $j>0$.
\ifnarrowcols  \else (Thus $\G_0$ is a $(n+1)\times(n+1)$ version of Equation~(\ref{eqn1}).)\fi

Rescale the payoffs of $\G$ to all lie in the range $[0\cdotp9,1\cdotp1]$.
Let $\G_1$ be a $(n+1)\times(n+1)$ game with strategies
$\{s^r_0,\ldots,s^r_n\}$ for the row player, and
$\{s^c_0,\ldots,s^c_n\}$ for the column player. Payoffs are as follows:

\begin{itemize}
\item $(s^r_0,s^c_0)$ results in payoffs $(0,-1)$ for the players.\footnote{The
two-component payoff vectors assign the first component
to the row player and the second component to the column player.}
\item $(s^r_0,s^c_j)$ for $j>0$ results in payoffs $(0,\frac{3}{4})$.
\item $(s^r_j,s^c_0)$ for $j>0$ results in payoffs $(-1,\frac{3}{4})$ for $j\not=k$,
      and $(-1,\frac{3}{4}+\delta)$ (for $\delta$ inverse polynomial in $n$) for $j=k$
\item The rest of $\G_1$ is a copy of $\G$ above.
\end{itemize}

\ifnarrowcols
\else
\[
\G_1 =
\begin{tabular}{c|cccc}
         & $s^c_0$  & $s^c_1\cdots s^c_{k-1}$ & $s^c_k$ & $s^c_{k+1}\cdots s^c_n$ \\ \hline
$s^r_0$  & $(0,-1)$ & $(0,\frac{3}{4})\cdots(0,\frac{3}{4})$ & $(0,\frac{3}{4}+\delta)$ & $(0,\frac{3}{4})\cdots(0,\frac{3}{4})$ \\
$s^r_1$  & $(-1,\frac{3}{4})$  &                         &         &       \\
$\vdots$ & $\vdots$ &                         &   $\G$  &       \\
$s^r_n$  & $(-1,\frac{3}{4})$  &                         &         &       \\
\end{tabular}
\]
\fi

Let $\G_t=(1-t)\G_0+t\G_1$.
The above payoffs have been chosen so that Nash equilibria satisfy:
in $\G_1$, players do not use $s^r_0$ or $s^c_0$;
in $\G_{0\cdotp6}$, players both have a proper mixture of $s^r_0$ and $s^c_0$
with their other strategies. Since $\G$'s payoffs were rescaled to lie in
$[0\cdotp9,1\cdotp1]$, $\Pr[s^r_0]$ and $\Pr[s^c_0]$ can be shown to lie in $[0\cdotp1,0\cdotp9]$,
which can be checked from the following payoff ranges for $\G_{0\cdotp6}$:
\[
\ifnarrowcols \begin{tabular}{@{}r@{}|c|c@{}|}
\else \begin{tabular}{r|c|c|} \fi
\multicolumn{1}{r}{}
 & \multicolumn{1}{c}{$s^c_0$} & \multicolumn{1}{c}{$s^c_1\ldots s^c_n$} \\
\cline{2-3}
 $s^r_0$ & $(0\cdotp4,-0\cdotp2)$    & $(0\cdotp4,0\cdotp45+\delta)$ \\
\cline{2-3}
 $s^r_1\ldots s^r_n$ & $(-0\cdotp6,0\cdotp85)$ & $([0\cdotp54,0\cdotp66],[0\cdotp54,0\cdotp66])$   \\
\cline{2-3}
\end{tabular}
\]
{\em Thus a continuous path of equilibria should at some stage allocate gradually less and less
probability to $s^r_0$ and $s^c_0$ as $t$ increases.}

\begin{observation}\label{subgameNE}
In any Nash equilibrium ${\cal N}$ of $\G_1$, the players assign probability 0 to
$s^r_0$ and $s^c_0$, and consequently ${\cal N}$ consists of a Nash
equilibrium of $\G$, restricting to strategies $s^r_j$, $s^c_{j'}$ for $j,j'>0$.
\end{observation}

Since $\G_t=(1-t)\G_0+t\G_1$, we can write $\G_t$ as
\ifnarrowcols
\else
\[
\begin{tabular}{c|cccc}
         & $s^c_0$  & $s^c_1\cdots s^c_{k-1}$ & $s^c_k$ & $s^c_{k+1}\cdots s^c_n$ \\ \hline
$s^r_0$  & $(1-t,1-2t)$  & $(1-t,\frac{3}{4}t)\cdots(1-t,\frac{3}{4}t)$ & $(1-t,(\frac{3}{4}+\delta)t)$ & $(1-t,\frac{3}{4}t)\cdots(1-t,\frac{3}{4}t)$ \\
$s^r_1$  & $(-t,1-\frac{1}{4}t)$  &                         &         &       \\
$\vdots$ & $\vdots$ &                         &   $t\G$  &       \\
$s^r_n$  & $(-t,1-\frac{1}{4}t)$  &                         &         &       \\
\end{tabular}
\]
\fi

The general idea is as follows. Consider the Browder path of equilibria
that begins from the unique equilibrium of $\G_0$ (where initially both
players play $s^r_0$, $s^c_0$). As $t$ increases, the players will start
to use the other strategies. At that stage, consider the distribution of
their mixed strategies restricted to $s^r_1,\ldots,s^r_n$ and
$s^c_1,\ldots,s^c_n$. These distributions will constitute a Nash
equilibrium of a version of $\G$ in which the column player receives a
small bonus for playing $s^c_k$. As $t$ increases to 1, the bonus
decreases continuously to 0, and we recover Observation~\ref{subgameNE}.
Now, recall from Definition~\ref{def:circgame} that the way
\cite{DGP,CDT} reduce graphical games to two-player games, is to
associate each player $v$ in the graphical game with two strategies in
the two-player game, both belonging to the same player. The division of
probability between those two strategies represents the probability that
$v$ plays 1. Consider $v_{\switch}$ now, corresponding
to $s^c_k$ and $s^c_{k+1}$. $v_{\switch}$ is, in the graphical
game, mildly incentivized to play 1, but for $t<1$ the $\delta$ in the
two-player game $\G_t$ pushes it the other way, towards 0. As a result,
a Nash equilibrium of $\G_t$ may simulate a Nash equilibrium of $\GG^+_t$ where $\pp[v_{switch}]\in(0,1)$.
As $t$ increases and the contribution from $\delta$ decreases, this process
corresponds to raising $\pp[v_{\switch}]$ continuously (but not monotonically) from 0 to 1.

\begin{lemma}\label{lem:gt}
Let ${\cal N}$ be a Nash equilibrium of $\G_t$ in which
$\Pr[s^r_0]<1$ and $\Pr[s^c_0]<1$.
Let ${\cal P}$ be the probability distributions over
$\{s^r_1,\ldots,s^r_n\}$ and $\{s^c_1,\ldots,s^c_n\}$ obtained by
taking each value $\Pr[s^i_j]$ (for $i\in\{r,c\}$, $1\leq j\leq n$)
and dividing it by $1-\Pr[s^i_0]$.

Then ${\cal P}$ is a Nash equilibrium of a version of $\G$ where the
column player receives a bonus of  $\delta\Pr[s^r_0]/(1-\Pr[s^r_0])$ for $s^c_k$.
\end{lemma}

\begin{proof}
In Nash equilibrium ${\cal N}$, $c$'s strategy $s^c_0$ contributes
the same quantity to each one of $r$'s strategies $s^r_1,\ldots,s^r_n$.
So the values $\Pr[s^r_1],\ldots,\Pr[s^r_n]$ must form a best response
to $c$'s mixed strategy from ${\cal P}$.

The column player receives a bonus $\delta t\Pr[s^r_0]$ specific to $s^c_k$, arising
from the possibility that row player plays 0. He also receives an
additional $\frac{3}{4}t$ for all strategies $s^c_j$ for $j>0$, but that
uniform bonus has no further effect on his preference amongst $s^c_1,\ldots,s^c_n$.

So in ${\cal N}$, $\Pr[s^c_1],\ldots,\Pr[s^c_n]$ is a best response to
a mixture of $\G$ weighted by $1-\Pr[s^r_0]$ and the probability
$\Pr[s^r_0]$ of a bonus $\delta\Pr[s^r_0]$ for playing $s^c_k$.
This is equivalent to a best response to a version of $\G$ with a
bonus of $\delta\Pr[s^r_0]/(1-\Pr[s^r_0])$ for playing $s^c_k$.
\end{proof}

Consider the path of equilibria connecting equilibrium $\NE_0$ of $\G_0$ to equilibrium $\NE_1$ of $\G_1$.
By Lemma~\ref{lem:gt} we can choose $\delta$ such that in any equilibrium of $\G_{0.5}$ we have $\Pr[s^c_{k+1}]=0$.
We also have that in any equilibrium of $\G_1$, $\Pr[s^c_k]=0$.
Consider the longest suffix of the path for which $t\geq 0.5$ for all games $\G_t$ that appear in that suffix.
The corresponding equilibria assign weight strictly less than 1 to $s^r_0$ and $s^c_0$, so Lemma~\ref{lem:gt}
may be used to recover corresponding equilibria of versions of $\G$ which in turn correspond to
versions of $\GG^+$ in which initially, $v_{switch}$ is incentivized to play 0,
and finally, $v_{switch}$ is incentivized to play 1.

\section{From Linear Tracing to the homotopies of van den Elzen-Talman, Herings-van den Elzen, and Herings-Peeters}\label{sec:extVdET-HP}

In the previous section, we showed the \PSPACE-completeness of finding
the Nash equilibrium of a two-player game that is associated with a
homotopy that uses a specific simple starting-game that is not derived
from the game of interest. In the
literature on homotopy methods, starting with Harsanyi~\cite{Harsanyi},
the starting-game is usually derived from the game of interest by
positing a prior distribution over the players' pure strategies, and
using a starting-game whose payoffs are the result of playing against
this prior distribution. 
%
In this section, we extend the result of Section~\ref{sec:lt} to handle 
these starting-games and thus obtain results for the  
Herings-van den Elzen~\cite{HE} and Herings-Peeters~\cite{HP1} algorithms, which 
use the same underlying homotopy, and the van den Elzen-Talman~\cite{HP} algorithm, 
which uses a different homotopy. All three algorithms have been shown under certain conditions to mimic 
the Harsayni-Selten linear tracing procedure. 
For each algorithm, we use the uniform distribution as the prior distribution, 
which is a natural choice. 

The van den Elzen-Talman algorithm uses a homotopy based on a starting mixed-strategy profile $v$.
Letting $\Sigma$ be the set of mixed-strategy profiles, let $\Sigma(t)$ be the
set of convex combinations $(1-t)\{v\}+t\Sigma$.
In the notation of~\cite{HP}, the van den Elzen-Talman algorithm ---restricted to the
two-player case--- uses the homotopy
\[
H(t,\sigma) = \beta^1_{\sigma^1(t)}(\sigma) \times \beta^2_{\sigma^2(t)}(\sigma)
\]
where for $i=1,2$, $\beta^i_{\sigma^i(t)}(\sigma)$ denotes the best responses of player $i$
to mixed strategy $\sigma$, restricted to $\Sigma(t)$.

\begin{theorem}
It is \PSPACE-complete to compute equilibria that result from the above van den Elzen-Talman homotopy.
\end{theorem}

\ifnarrowcols
\else
\begin{proof}(sketch)
It can be checked that the algorithm uses polynomial space.
For the hardness, we reduce from {\sc Linear tracing}; consider a game ${\cal G}$ for
which we seek an equilibrium that results from starting with ${\cal G}_0$
\ifnarrowcols
as in Section~\ref{sec:lt}.
\else
of the form of~(\ref{eqn1}).
\fi
Suppose we take a game $\G$ from Section~\ref{sec:lt} and give each player an additional
strategy as follows. Let $s^r_n$ and $s^c_n$ be the new strategies, for the row and
column player respectively. $s^r_n$ has a payoff of $-10$ for the row player, regardless
of how the column player plays (thus, $s^r_n$ is dominated by all the other strategies).
The payoffs to the column player are chosen in such a way that, if in fact the row
player uses the uniform distribution over $s^r_0,\ldots,s^r_n$, then the column player's
payoffs will be 1 for $s^c_0$ and 0 for $s^c_j$, for $j>0$. The new strategy $s^c_n$
has a similar definition. Note that the new payoffs are at most $n$ in absolute value.
These new strategies ensure that we have the desired $\G_0$ of Section~\ref{sec:lt}, when
we restrict each player to his first $n$ strategies. Strategies $s^r_n$ and $s^c_n$ are
not used in any Nash equilibrium of $\G_t$, since they are strictly dominated for all $t$.

We let $v$ be the uniform distribution. In ${\cal G}_t$, the row plays a mixture
$(1-t)v+t\sigma^r_t$ while the column player plays $(1-t)v+t\sigma^c_t$, where
$\sigma^r_t$ and $\sigma^c_t$ are mixed strategies whose support do not include
$s^r_n$ and $s^c_n$, so they constitute a Nash equilibrium of a version of ${\cal G}$
in which there is a bonus to play $s^r_0$ and $s^c_0$. This bonus drops continuously
to zero, so it is equivalent to the linear-tracing homotopy.
\end{proof}



The algorithms of Herings-van den Elzen~\cite{HE} and Herings-Peeters~\cite{HP1} are based on
an identical homotopy and differ only in the numerical technique used to follow the homotopy path.
We can show using essentially the same construction as above that is is \PSPACE-hard to compute 
the equilibria found by these homotopies. To do so we can again construct a starting game by giving the 
row/column players new strategies $s^r_n$ and $s^c_n$ chosen to have low payoffs to the row 
(respectively, column) players, but whose payoffs to the opponent are chosen such that if either player
played the uniform distribution, the opponent would receive a higher payoff for
his first strategy (either $s^r_0$ or $s^c_0$) than the others, which would all 
receive the same (lower) payoffs.
\fi

\section{From Linear Tracing to Lemke-Howson}
\label{sec:extLH}

The Lemke-Howson (L-H) algorithm is an important and rich research subject
in and by itself within Game Theory; for the purposes of this reduction,
it is helpful to take a point of view that considers the L-H algorithm
as a homotopy~\cite{HP}, where an arbitrary strategy (the one whose label is
dropped initially) is given a large ``bonus'' to be played, so that
the unique equilibrium consists of that strategy together with its
best response from the other player; the homotopy arises from reducing
that bonus continuously to zero. 
\begin{theorem}\label{LHtheorem}
It is \PSPACE-complete to find {\em any} of the solutions of a 2-player game that
are constructed by the Lemke-Howson algorithm.
\end{theorem}

The remainder of this section proves Theorem~\ref{LHtheorem}, the hardness being established
by a reduction from the graphical game problem of Proposition~\ref{prop:gg}, extending the
ideas of the reduction for {\sc Linear tracing} (Theorems~\ref{thm:lt},~\ref{thm:ltdetails}).
A new technical challenge here is that the choice of initially dropped label
results in $2n$ alternative homotopy paths, and we must ensure that any of the (up to) $2n$ solution
can encode the single solution to some instance of {\sc Linear tracing}.

Suppose that some strategy has been given this ``L-H bonus'', and a Browder path of Nash
equilibria is obtained from reducing that bonus to zero.  As before
let $t\in[0,1]$ be a parameter that denotes the distance from the
starting game of the homotopy to the game of interest, so that $1-t$
is a multiplicative weight for the bonus in intermediate games.
Consider the Browder path. It is piecewise linear, a topologically
well-behaved line. Let $T\in[0,1]$ parameterize points along the
Browder path --- an equilibrium ${\cal N}_T$ is the one that is a
fraction $T$ of the distance along the path (starting at the version
of the game with the L-H bonus).  So, multiple values of $T$
can correspond to the same value of $t$.  Here we mostly focus on $T$
rather than $t$.

The following construction addresses the issue that an arbitrary strategy may receive the L-H bonus.
We embed two copies of a circuit-encoding game $\G$ (Definition~\ref{def:circgame}) into a game
instance for the Lemke-Howson algorithm. At least one of those copies of $\G$ will not contain the strategy
that receives the L-H bonus. The L-H homotopy, restricted to that copy of $\G$,
will simulate the homotopy of Section~\ref{sec:lt}.

In Figure~\ref{LHdiagram}, ${\cal G}$ denotes a circuit-encoding $n\times n$ game (note the two copies)
whose payoffs have been rescaled to lie in the interval $[0\cdotp4,0\cdotp6]$. ${\cal G}$ is
assumed to have an associated graphical game with two ``switch'' players $v^r_{\switch}$, $v^c_{\switch}$
that affect the equilibria of ${\cal G}$ according to Corollary~\ref{cor:switches}.
They will correspond to the first pair of each of ${\cal G}$'s players' strategies
$(s^r_0,s^r_1)$ and $(s^c_0,s^c_1)$ such that,
\begin{itemize}
\item if both $\pp[v^r_{\switch}]=1$ and $\pp[v^c_{\switch}]=1$, ${\cal G}$'s equilibrium encodes a solution
to an {\sc End of the line} instance that is efficiently encoded by ${\cal G}$;
\item if either $\pp[v^r_{\switch}]=0$ or $\pp[v^c_{\switch}]=0$, ${\cal G}$ encodes the ``basic''
Brouwer-mapping function;
\item if we add a bonus to the row player for his first strategy $s^r_0$ that is
less than some threshold $\tau$, it will result in $\Pr[s^r_0]=0$ and hence $\pp[v^r_{\switch}]=1$,
and similarly for the column player with respect to $s^c_0$ and $v^c_{\switch}$.
(We will see that such bonuses occur, and they decrease at $T\longrightarrow 1$.)
\end{itemize}

\noindent{\bf Notation.} $A,B,C,D$ and $A',B',C',D'$ denote sets of the players' strategies as
shown in Figure~\ref{LHdiagram}. In the context of a mixed-strategy profile,
$\Pr[C]$ denotes the probability that the column player uses $C$; $\Pr[A]$ that
he chooses an element of $A$, and so on.
Let $X(T) = \Pr[C]+\Pr[D]+\Pr[C']+\Pr[D']$, a function of distance along
the Browder path. We note the following facts
\begin{itemize}
\item $X(0)\geq 1$ (if, say, a column player strategy receives the L-H bonus, then the row
player will play some pure best response, either $C'$ or $D'$; so $\Pr[C']=1$ or $\Pr[D']=1$.)
\item $X(1)\leq \frac{1}{25}$ (shown in Lemma~\ref{lem:endpath})
\end{itemize}
together with the key observation that $X(T)$ is a continuous function
of $T$, implying:
\begin{observation}
For some $T'\in[0,1]$, $X(T')=\frac{1}{4}$, and for $T>T'$, $X(T)<\frac{1}{4}$.
\end{observation}

Let $\bar{\cal G}$ be the copy of ${\cal G}$ that does not contain the strategy that receives
the L-H bonus. (If one of $C$, $D$, $C'$ or $D'$ receive the L-H bonus, then $\bar{\cal G}$ may be
either copy of ${\cal G}$.)

For any $X$, at least one player $p$ has an additional bonus at least $X/4$ to play $s^p_0$ in $\bar{\cal G}$
(suppose for example $\Pr[C]+\Pr[D]\geq X/2$ and $p$ is the row player; Figure~\ref{LHdiagram} awards additional $e=1$ to $p$
when $C$ or $D$ is played). But neither player's bonus exceeds $X/2$.
As $T$ increases from $T'$ to 1, $X(T)$ goes down from $\frac{1}{4}$ to at most $\frac{4}{M}$.
We will establish that when $X(T)=\frac{1}{4}$, ${\cal
N}_T$ contains a solution to a ``biased'' version of $\bar{\cal G}$
where one of the players' first strategies (i.e. $s^r_0$ or $s^c_0$) has an additional bonus
(enough to ensure $\Pr[s^r_1]=0$ and $\pp[v^r_{\switch}]=0$, in the case of the row player).
Furthermore, when $T=1$, we have that ${\cal N}_T$ contains a solution
to $\bar{\cal G}$, only with smaller biases. These biases are associated with ``switch''
strategies in the graphical game associated with ${\cal G}$.

Let $T'$ be the largest value of $T$ where $X(T)$ is large enough that one of the
bonuses sets $\Pr[s^p_0]=1$ in $\bar{\cal G}$ (for $p\in\{r,c\}$). Between $T'$ and
$T=1$ we pass through a continuum of equilibria where $\Pr[s^p_0]$ changes from 1 to 0;
equivalently $\pp[v^p_{\switch}]$ changes from 0 to 1,
and the resulting equilibrium at $T=1$ corresponds to a solution to {\sc Oeotl}.

\begin{lemma}
Let ${\cal N}_T$ be a solution of ${\cal G}_T$. If $\bar{\cal G}$ is the bottom
right-hand copy of ${\cal G}$ in Figure~\ref{LHdiagram}, then if the distributions
over $B$ and $B'$ are normalised to 1, we have a Nash equilibrium of a game
$\hat{\cal G}$ where the row player has an additional bonus of
$e(\Pr[C]+\Pr[D])/\Pr[B']$ to play his first strategy $s^r_0$, and the column player
has an additional bonus of $e(\Pr[C']+\Pr[D'])/\Pr[B]$ to play his first strategy $s^c_0$.
\end{lemma}

By symmetry, a similar result also holds in the case that $\bar{\cal G}$ is the
top right-hand copy of the ${\cal G}$.

\begin{proof}
Payoffs to the row player are unaffected by the column player's distribution over
$A$. Meanwhile, $C$ and $D$ lead to an additional bonus of $e$ (weighted by the
probability that $C$ and $D$ are used by the column player) for the row player to
use the top row of $B'$.
\end{proof}

\begin{lemma}\label{lem:endpath}
At $t=1$ (equivalently, $T=1$) we have in any Nash equilibrium, that $\Pr[C]\leq\frac{1}{M}$,
$\Pr[D]\leq\frac{1}{M}$, $\Pr[C']\leq\frac{1}{M}$ and $\Pr[D']\leq\frac{1}{M}$.
Since $M\geq 100$ we have $X(1)\leq \frac{1}{25}$.
\end{lemma}

\begin{figure}
\begin{center}
\ifnarrowcols
\else
\includegraphics[width=6.0in]{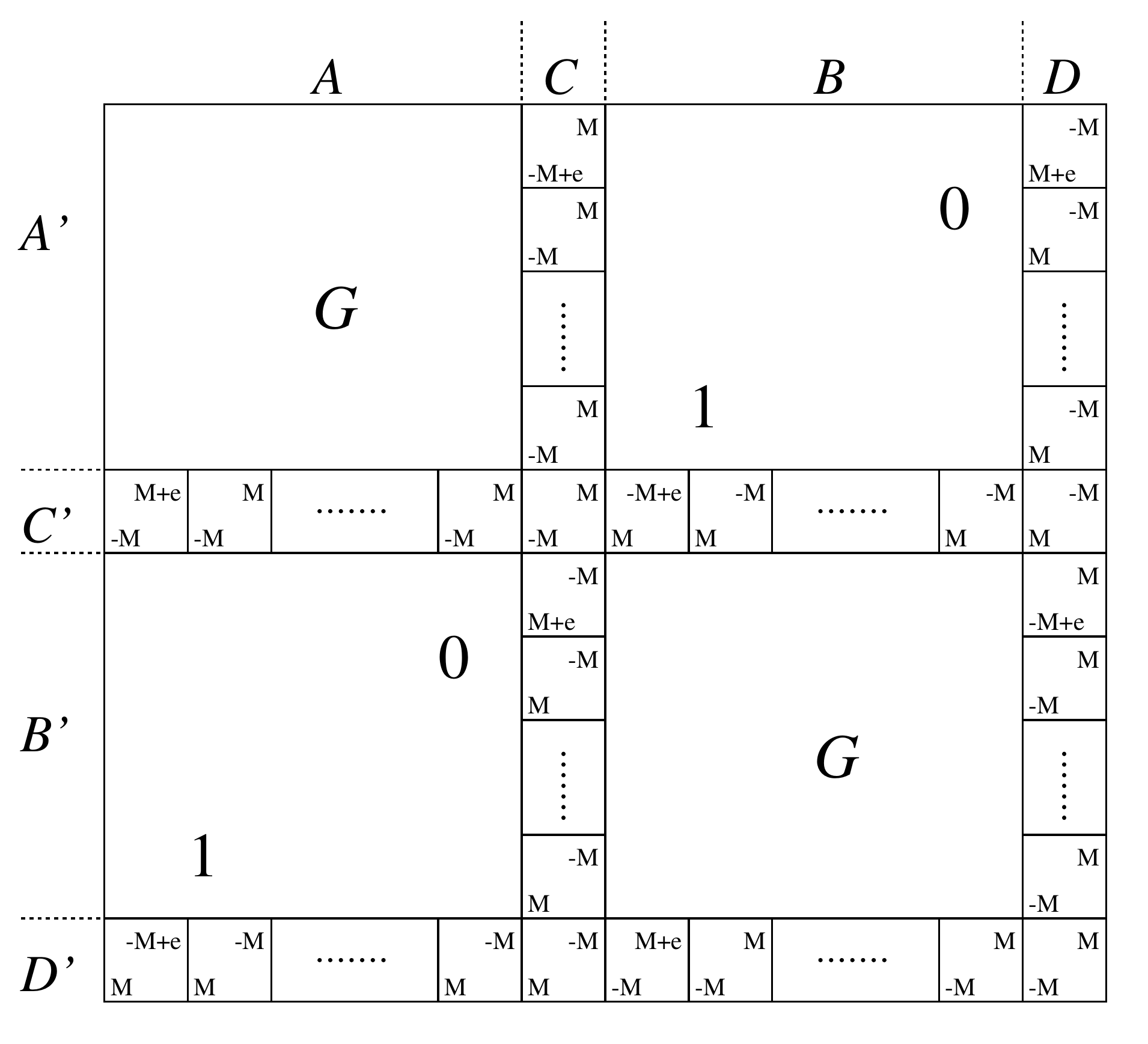}
\fi
\end{center}
\caption{The game has 2 copies of $n\times n$ game ${\cal G}$ embedded in the top-left
and bottom-right regions, with payoff rescaled to $[0\cdotp4,0\cdotp6]$.
In the top-right and bottom-left regions are copies
of a $n\times n$ game that give the column player a payoff of
0 and the row player a payoff of 1.\newline
Each of $A$, $B$, $A'$, $B'$ denotes a set of $n$
strategies. $C$, $D$, $C'$ and $D'$ are individual strategies.\newline
In the proofs we put $M=1000$, $e=1$.}\label{LHdiagram}
\end{figure}

\ifnarrowcols
\else
\begin{proof} of Lemma~\ref{lem:endpath}.
We give the proof that $\Pr[C]\leq\frac{1}{M}$; by symmetry the other claims are similar.

Assume for contradiction that $\Pr[C]>\frac{1}{M}$.
We know that at least one strategy from $A\cup B$ gives positive payoff, since the
places where they can obtain a negative payoff are ``equal and opposite''.
Indeed, it can be checked that the payoff to at least one member of $A$ and $B$
is at least $\frac{1}{2}\min\{0\cdotp4,e\}$.
The (column player's) payoffs from $C$ and $D$ must sum to zero, so if $\Pr[C]>0$
then $\Pr[D]=0$ ($D$ gets negative payoff and is a strictly worse response $A$ or $B$.)
A similar argument for the row player's payoffs establishes that one or
both of $C'$ and $D'$ gets zero probability.

Given that $\Pr[C]>\frac{1}{M}$ we can deduce that $\Pr[A']=0$ due to being a worse
response than $B'$: $C$ contributes at least $\frac{2M}{M}=2$ to $\payoff(B')-\payoff(A')$;
$A$ contributes a positive amount; $D$ has zero probability so contributes nothing;
$B$ contributes $\geq -0\cdotp6$.
We noted above that $\Pr[C']=0$ or $\Pr[D']=0$ (or both). Consider two cases:

{\bf Case 1:} $\Pr[D']=0$. Deleting strategies with probability zero, we are left with
the following structure:

\begin{center}
\begin{tabular}{r|c|c|c|}
\multicolumn{1}{r}{}
 & \multicolumn{1}{c}{$A$} & \multicolumn{1}{c}{$C$}  & \multicolumn{1}{c}{$B$} \\
\cline{2-4}
 $B'$ & $(1,0)$    & $([M,M+e],-M)$ & $([0\cdotp4,0\cdotp6],[0\cdotp4,0\cdotp6])$ \\
\cline{2-4}
 $C'$ & $(-M,[M,M+e])$ & $(-M,M)$   & $(M,[-M,-M+e])$  \\
\cline{2-4}
\end{tabular}
\end{center}


Comparing $C$ with $A$, we need $\Pr[C']=1$ to avoid $A$ being a better response than
$C$ (which is supposed to have positive probability $>\frac{1}{M}$).
If $\Pr[C']=1$, $B$ is a worse response than the others, but when the column player
uses only $A$ and $C$, $C'$ has much lower payoff than $B'$, contradicting
assumption that $\Pr[C']$ is positive. This leaves us with Case 2:

{\bf Case 2:} $\Pr[C']=0$. We get, after deleting zero-probability strategies,

\begin{center}
\begin{tabular}{r|c|c|c|}
\multicolumn{1}{r}{}
 & \multicolumn{1}{c}{$A$} & \multicolumn{1}{c}{$C$}  & \multicolumn{1}{c}{$B$} \\
\cline{2-4}
 $B'$ & $(1,0)$    & $([M,M+e],-M)$ & $([0\cdotp4,0\cdotp6],[0\cdotp4,0\cdotp6])$ \\
\cline{2-4}
 $D'$ & $(M,[-M,-M+e])$ & $(M,-M)$   & $(-M,[M,M+e])$  \\
\cline{2-4}
\end{tabular}
\end{center}

Here the contradiction is immediate since $B$ is a strictly better response than $C$,
preventing $\Pr[C]>0$.
\end{proof}
\fi


\begin{lemma}\label{lem:LBweight}
Assume that $e\leq 1$ in Figure~\ref{LHdiagram} and that $M\geq 100$.
Suppose that $X(T)\leq\frac{1}{4}$. Then $\Pr[A]\geq\frac{1}{10}$, $\Pr[B]\geq\frac{1}{10}$,
$\Pr[A']\geq\frac{1}{10}$, $\Pr[B']\geq\frac{1}{10}$.
\end{lemma}

\ifnarrowcols
\else
\begin{proof} 
We need to consider two cases in detail: case 1 assumes that an element of $A'$
received the L-H bonus and case 2 assumes that $C$ received the bonus. All
other possibilities are essentially the same as these, by symmetry.

\smallskip
{\bf Case 1.} Suppose first that a strategy from $A'$ has been given the L-H bonus, and
we are at a Nash equilibrium where $X=\frac{1}{4}$.

First we prove that the row player strategies satisfy
$\Pr[A']\geq\frac{1}{10}$, $\Pr[B']\geq\frac{1}{10}$.
Since $X\leq\frac{1}{4}$, we have $\Pr[A']+\Pr[B']\geq \frac{3}{4}$.
Suppose for a contradiction that $\Pr[B']<\frac{1}{10}$, so that $\Pr[A']\geq 0\cdotp65$.
Then (for $e\leq 1$),
$C$ is the unique best response for the column player and hence $\Pr[C]=1$.
This implies $X\geq 1$, contradicting the assumption that $X\leq\frac{1}{4}$.
Similarly, if $\Pr[A']<\frac{1}{10}$ then $\Pr[B']\geq 0\cdotp65$, then $D$
is the column player's unique best response, hence $\Pr[D]=1$, again contradicting
$X\leq\frac{1}{4}$.

\smallskip
Next we prove that the column player strategies satisfy
$\Pr[A]\geq\frac{1}{10}$, $\Pr[B]\geq\frac{1}{10}$.
Suppose for a contradiction that $\Pr[B]<\frac{1}{10}$, so that $\Pr[A]\geq0\cdotp65$.
Then for the row player, $D'$ is a better response than $B'$ and $C'$ (and regarding $A'$, some
strategy from $A'$ received the L-H bonus, so we do not claim $A'$ is
suboptimal).
If $B'$ is not a best response, so $\Pr[B']=0$, since $X\leq\frac{1}{4}$ we have $\Pr[A']\geq\frac{3}{4}$.
Then $C$ is strictly better than $A$, contradicting $\Pr[A]\geq0\cdotp65$.
Alternatively suppose that $\Pr[A]<\frac{1}{10}$, so that $\Pr[B]\geq0\cdotp65$.
Then $C'$ is a better response than $B'$ and $D'$, since $\payoff(C')\geq(0\cdotp65-0\cdotp35)M=0\cdotp3M$;
$\payoff(B')\leq\frac{1}{4}M$ since $\Pr[C]\leq\frac{1}{4}$ by
assumption that $X\leq\frac{1}{4}$; $\payoff(D')$ is negative.
With $D'$ and $B'$ eliminated, $C$ is a strictly better response than~$B$, contradicting $\Pr[B]>0$.

\medskip
{\bf Case 2.} Suppose alternatively that it was strategy $C'$ that received the L-H bonus.

We show first that the row player's strategies satisfy
$\Pr[A']\geq\frac{1}{10}$, $\Pr[B']\geq\frac{1}{10}$.
Suppose $\Pr[A']<\frac{1}{10}$, so that $\Pr[B']\geq0\cdotp65$. $D$ is a better response than
$A$ and $B$. But from that it follows that
$\Pr[D]+\Pr[C]=1$, contradicting $X=\frac{1}{4}$.
Suppose $\Pr[B']<\frac{1}{10}$, so that $\Pr[A']\geq0\cdotp65$.
$C$ is a better response than $A$ and $B$, so $\Pr[C]+\Pr[D]=1$ contradicting $X=\frac{1}{4}$.

\smallskip
Next we show that the column player's strategies satisfy
$\Pr[A]\geq\frac{1}{10}$, $\Pr[B]\geq\frac{1}{10}$.
Suppose $\Pr[B]<\frac{1}{10}$, so $\Pr[A]>0\cdotp65$. $D'$'s payoff is
greater than $0\cdotp3M$ while $A'$ and $B'$ have payoff at most
$\frac{1}{4}M+1$ so $D'$ is a better response than $A'$ and $B'$, contradicting
$X\leq\frac{1}{4}$.
Suppose $\Pr[A]<\frac{1}{10}$, so $\Pr[B]>0\cdotp65$. $C'$ is a better response than
$A'$ and $B'$ even ignoring the L-H bonus.
\end{proof}
\fi

At $X(T)=\frac{1}{4}$ we have that at least one of $\Pr[C]$, $\Pr[D]$,
$\Pr[C']$, $\Pr[D']$ is at least $\frac{1}{16}$, while at the end
of the Browder path, we know that all these quantities are
at most $\frac{1}{100}$. We set the switch threshold probability to be
somewhere between these, but we have to use lower bounds on $\Pr[A]$,
$\Pr[B]$, $\Pr[A']$, $\Pr[B']$ at $t=1$ and upper bounds on these
at $X=\frac{1}{4}$ (as well as lower bounds on these at $X=\frac{1}{4}$
to ensure that a Nash equilibrium of the ``biased game'' is being
encoded).

\medskip
Finally, we need to show that there exists $\tau$ such that the bonus from at least one switch strategy in $\bar{\cal G}$ changes
continuously above $\tau$ to below it, while the bonus
for the other switch strategy ends up below $\tau$, thus initially, at least one
value of $\pp[v^r_{\switch}]$ and $\pp[v^c_{\switch}]$ is zero, but at the end both
evaluate to 1. This needs to take into account the variable amount of probability allocated
to the strategies in $\bar{G}$, since that affects the impact of the bonuses on $s^p_0$.

For any $T\in[T',1]$ the weight assigned by each player
to $\bar{\cal G}$'s strategies is at least $\frac{1}{10}$ by Lemma~\ref{lem:LBweight},
so that the bonus for player $p$ to play $s^p_0$, falls by a larger factor than the
probability that $\bar{G}$ is played. That means that $\tau$ can indeed be chosen as required.

\section{Discussion and Open Problems}

Should a more general result be obtainable? For example, perhaps it should be possible to
identify general classes of ``path-following algorithms'' that include the ones we analyzed here, for which
it is \PSPACE-complete to compute their output.
\ifnarrowcols
\else
A potential obstacle is that such a general result may subsume the question of whether
the 2-dimensional analogue of {\sc Oeotl} is \PSPACE-complete (i.e. consider
the PPAD-complete problem {\sc 2d-Sperner}~\cite{CD}; suppose we ask for
the trichromatic triangle identified in the proof of Sperner's Lemma.) In 2 dimensions,
the gadget that is used to allow ``edges'' to cross each other, rearranges the structure
of those edges, such that the corresponding solutions to {\sc End of the line} are the
same, but not the unique solution to {\sc Oeotl}. Generally, there are many
ways to modify the edges of a given $(S,P)$-graph so that the degree-1 vertices are
unchanged, but the structure of the graph is in other respects completely different.

Von Stengel et al~\cite{SET} use a tracing procedure to solve extensive
two-person games, and they obtain a {\em normal form perfect}
equilibrium by starting from a completely mixed starting vector. What is
the complexity of computing a normal form perfect solution using this
(or other) methods? (They note (\cite{SET}, p.~707) that on strategic-form
games this procedure mimics the linear tracing procedure of~\cite{HS}.)
\fi

\ifnarrowcols
\bibliographystyle{IEEEtranS}
\begin{small}
\else
\fi

\ifnarrowcols
\end{small}
\else
\fi

\ifnarrowcols
\else
\section{Appendix: exponentially many changes of direction}
\label{sec:exp-changes}

We give an outline of how to modify our first construction, so as to
show that in following a sequence of approximate fixpoints of Brouwer
functions, or equilibria of games, $t$ may have to change direction
exponentially many times, and furthermore, oscillate between values whose
difference is bounded away from zero.

$\F_0$ shall be the same as in Theorem~\ref{thm:pspace-brouwer}. We construct
a modified form of $\F_1$, which we will call $\F^m_1$, as follows.
We identify two subsets of the cubelets $\cube_{ijk}$, $R_1$ and $R_2$, defined as
$R_1=\{\cube_{ijk}~:~10\leq k\leq 12\}$, and
$R_2=\{\cube_{ijk}~:~20\leq k\leq 22\}$. These subsets are ``thin layers'' of cubelets
that are perpendicular to the $z$-axis. Now we define how $\F^m_1$ behaves on points
at the centers of cubelets. For $x$ at the center of $\cube_{ijk}$,
\[
\begin{tabular}{c}
if $\cube_{ijk}\in R_1$, $\F^m_1(x)-x=10(\F_1(x)-x)$,\\
if $\cube_{ijk}\in R_2$, $\F^m_1(x)-x=\frac{1}{10}(\F_1(x)-x)$,\\
otherwise, $\F^m_1(x)=\F_1(x)$.
\end{tabular}
\]

For points not at the centers of cubelets, $\F^m_1$ shall interpolate between
the values at the nearest cubelet centers, using the same general approach as
$\F_1$. Let $\F^m_t=(1-t)\F_0+t\F^m_1$.

Let $D$ be the set of cubelets $\cube_{ijk}$ for which
$f_0(\cube_{ijk})\not=f_1(\cube_{ijk})$, so $D$ is the region where
fixpoints of $\F^m_t$ may exist. Let $D^-$ be the connected component of
$D$ which contains the cubelet $\cube_{1,1,1}$, so $D^-$ is the region
within which fixpoints of $\F^m_t$ will occur on the homotopy path.

By construction, $D^-$ has the property that $D^-\setminus R_1$ (and similarly
$D^-\setminus R_2$) has exponentially many connected components. This
follows from the construction of~\cite{DGP} from which $D$ is derived;
$D$ simulates a $(S,P)$-graph by, for each edge of the $(S,P)$-graph, including
a long sequence of cubelets that passes through both $R_1$ and $R_2$.

Now consider points in $D^-\cap R_1$. The claim is that fixpoints $x$ of
$\F^m_t$ for which $x\in D^-\cap R_1$, satisfy $t\leq \frac{1}{4}$, and that
for fixpoints $x$ of $\F^m_t$ with $x\in D^-\cap R_2$ we have $t\geq \frac{3}{4}$.
The general idea (in the first case; the second case is similar) is that
$\F^m$ will map points $z$ in $R_1$ to points $z'$
for which at least one component of $z'-z$ is greater than $10\alpha$.
Hence for $t\geq \frac{1}{4}$, $\F^m_t$ raises the value of this component
(the positive contribution from $\F^m_1$ exceeds the negative contribution
from $\F_0$) and prevents it from being a fixpoint.

The homotopy path must pass through this long sequence of regions that require
$t>\frac{3}{4}$ or else $t<\frac{3}{4}$. Moreover, the two types of
regions alternate, so we establish the following result:

\begin{theorem}
For continuous functions defined using arithmetic circuits, the sequence
of fixpoints along the path given by the linear homotopy $(1-t)\F_0 + t\F_1$
has exponentially many alternations of the value of $t$. 
\end{theorem}

We obtain the following corollary:

\begin{corollary}
For graphical or two-player games, suppose $\G_0$ is a game that assigns each player
a dominating strategy, and $\G$ is an arbitrary game. The linear tracing procedure
for the homotopy $(1-t)\G_0+t\G$ will, in the worst case, have exponentially
many reversals of $t$.
\end{corollary}

The corollary follows since, the way we represent Brouwer functions parameterized
by $t$ in terms of games parameterized by $t$, does not change the value of $t$. Since
we did not change $\F_0$, the associated game $\G_0$ is the same ``dominating
strategy'' game of Section~\ref{sec:lt}.

\section{Appendix: polynomially small error}
\label{sec:polyerror}

We can use the machinery of Chen et al.~\cite{CDT} so that when we talk about the hardness
of finding an $\epsilon$-fixpoint, $\epsilon$ is allowed to be inverse polynomial
rather than inverse exponential. This is achieved by using the snake embeddings of~\cite{CDT}.

\begin{paragraph}{Snake-embeddings}
A snake embedding reduces a low-dimensional Brouwer-mapping function having $2^n$ of cubelets
in each dimension, to a $\Theta(n)$ dimensional bmf having $O(1)$ cubelets in each dimension, in such
a way that panchromatic vertices of the high-dimensional bmf efficiently encode panchromatic vertices
of the low-dimensional bmf. The reduction can be decomposed into a sequence of $\Theta(n)$ iterations,
in which at each iteration, the number of cubelets along some axis is reduced by a constant factor,
and we acquire an additional axis having $O(1)$ cubelets (in~\cite{CDT} it is in fact 8 cubelets).
(Intuitively, the space is folded a constant number of times and gains thickness along the new dimension.)
\end{paragraph}

\begin{paragraph}{A specific snake-embedding, and some notation}
We consider a snake-embedding of a 3-dimensional bmf $f_B$ of the type of Proposition~\ref{prop:graph}.
Initially the colors are $\{0,1,2,3\}$; let $c_i$ denote the new color at the $i$-th iteration and let $s$
be the number of iterations required to reduce to 8 the number of cubelets along each axis.
\end{paragraph}

Let $n=4+c_s$ be the dimension of the new snake-embedding.

\begin{definition}
Let $\cube^n$ be the unit $n$-dimensional cube. Partition $\cube^n$ into
cubelets $\cube^n_{\bf v}$ where ${\bf v} \in\{0,1,...,7\}^n$ represents
a cubelet of edge length $1/8$. A {\em Brouwer-mapping circuit} maps
each such cubelet to one of the colors $\{0,1,2...,n\}$. Again, a bmf
should be polynomial-time computable, and map exterior cubelets to color
$i$ for cubelets whose $i$-th coordinate contains the first 0 (when
${\bf v}$ contains a 0), otherwise color 0. Other types of bmf that
correspond to Definition~\ref{def:bmf} are defined analogously.
\end{definition}

The high-dimensional bmf can be computed by a Brouwer-mapping circuit $B'$ that is polynomial in the size of $B$.
Let $f_{B'}$ be the function computed.
The challenge is to implement $f_{B'}$ using an arithmetic circuit that is
polynomial in the size of the circuit that computes the bmf, and computes a Lipschitz continuous function.
The simplicial-decomposition technique of Theorem~\ref{thm:bmf2circ} no longer works when we
move to non-constant dimension, since the number of simplices per cubelet is exponential in the dimension.

\begin{observation}
After iteration $i$, we have a $(i+3)$-dimensional bmf in which $c_i$ becomes the ``background color''
corresponding to color 0 in the original 3-dimensional instance.

The cubelets having colors $\{0,1,2,3,c_1,\ldots,c_{i-1}\}$ are mapped to cubelets in the $(i+3)$-dimensional
instance in such a way as to have the same neighborhood structure, but with some duplication at the
folds of the embedding.
\end{observation}

\begin{paragraph}{The continuous implementation}
Define $\F_{B'}:\cube^n\longrightarrow\cube^n$ as follows. If $x$ lies at the center of a cubelet
(of length $1/8$), letting $j=f_{B'}(x)$, $\F_{B'}(x) = x+\delta_j$, where $\delta_j=(-\alpha,-\alpha,\ldots,-\alpha)$
if $j$ is the background color $c_s$, otherwise $\delta_j=(0,0,\ldots,0,\alpha,0,0,\ldots,0)$ where the position
of the non-zero entry depends on the color $j$, and is chosen to satisfy the boundary conditions of a bmf.

If $x$ does not lie at the center of a cubelet, we claim that for each color $j$, we can efficiently compute the $L_\infty$ distance
from $x$ to the closest center of a cubelet having color $j$, using a linear arithmetic circuit. Let $d_j(x)\in[0,1]$ be this distance.
Let $\lambda_j(x) = \max(0,(\frac{1}{10}-d_j(x)))$.
Then define $\F_{B'}(x) = x + \sum_j \delta_j.\lambda_j(x)$. In that expression for $\F_{B'}(x)$, $\delta_j$ is a constant vector,
so we are not multiplying two computed quantities together (which is disallowed in a linear arithmetic circuit).
\end{paragraph}

\begin{paragraph}{Why it works}
If we are not within distance $\frac{1}{10}$ of a panchromatic vertex, then $\lambda_j=0$ where $j$ is one of the missing colors.
However, there is some $j'$ for which $\lambda_{j'}>\frac{1}{16}$. The choice of the vectors $\delta_j$ ensure that
$|\F_{B'}(x)-x|\geq \frac{1}{16}$.

Consequently any approximate fixpoint of $\F_{B'}$ is close to a panchromatic vertex of $f_{B'}$. 
To show that it is close enough to permit that panchromatic vertex to be efficiently reconstructed from the coordinates
of the fixpoint, it is easiest to assume that there are 24 rather than 8 cubelets along each axis, with the original
cubelets having been divided into 27 smaller ones all having the same color. Then an approximate fixpoint can be
assumed to lie within $\frac{1}{30}$ of a panchromatic vertex.

We also need to point out that the high-dimensional bmf $f_{B'}$ has a path of $\{0,1,2,3,c_1,\ldots,c_{s-1}\}$-chromatic cubelets
which simulates the path of $\{1,2,3\}$-chromatic cubelets in the bmf $f_B$. Thus, the one obtained by following the path in $f_{B'}$,
encodes the one obtained by following the corresponding path in $f_B$.
\end{paragraph}

By way of a final remark, it is necessary for us to make a snake embedding of our 3D
graph into higher dimension, rather than (as in~\cite{CDT}) reduce from
the 2D version of the problem~\cite{CD}. This is because the
PPAD-completeness of {\sc 2D Sperner}~\cite{CD} is a reduction that
alters the structure of the {\sc End of the line} graph being encoded,
and so would not (in an obvious way) apply in a reduction from {\sc Oeotl}.

\fi

\end{document}